%% file: main.tex
\documentclass[fullpage,11pt]{article}

\usepackage{amsthm}
\usepackage{graphicx} % support the \includegraphics command and options
\usepackage{array} % for better arrays (eg matrices) in maths

\usepackage{amsmath, amssymb, amsfonts, verbatim}
\usepackage{hyphenat,epsfig,subfigure,multirow}

\usepackage[compact]{titlesec}

\usepackage[usenames,dvipsnames]{xcolor}
\usepackage[ruled]{algorithm2e}

\usepackage{tcolorbox}
\tcbuselibrary{skins,breakable}
\tcbset{enhanced jigsaw}

\usepackage[textwidth=0.8in,textsize=scriptsize]{todonotes}

\definecolor{DarkRed}{rgb}{0.5,0.1,0.1}
\definecolor{DarkBlue}{rgb}{0.1,0.1,0.5}

\usepackage{nameref}
\definecolor{ForestGreen}{rgb}{0.1333,0.5451,0.1333}
\definecolor{Red}{rgb}{0.9,0,0}
\usepackage[linktocpage=true,
	pagebackref=true,colorlinks,
	linkcolor=DarkRed,citecolor=ForestGreen,
	bookmarks,bookmarksopen,bookmarksnumbered]
	{hyperref}

\usepackage{bm}
\usepackage{url}
\usepackage{xspace} 
\usepackage[mathscr]{euscript}

\usepackage{mdframed}

\usepackage[noend]{algpseudocode}
\makeatletter
\def\BState{\State\hskip-\ALG@thistlm}
\makeatother

\usepackage{cite}
\usepackage{enumerate}

\usepackage[margin=1in]{geometry}

\newtheorem{theorem}{Theorem}
\newtheorem{lemma}{Lemma}[section]
\newtheorem{proposition}[lemma]{Proposition}
\newtheorem{corollary}[theorem]{Corollary}
\newtheorem{claim}[lemma]{Claim}
\newtheorem{fact}[lemma]{Fact}

\newtheorem{definition}{Definition}

\newtheorem{problem}{Problem}
\newtheorem{remark}[lemma]{Remark}

\newtheorem*{claim*}{Claim}
\newtheorem*{proposition*}{Proposition}
\newtheorem*{lemma*}{Lemma}
\newtheorem*{problem*}{Problem}

\newtheorem{mdresult}{Result}
\newenvironment{result}{\vspace{-0.10cm}\begin{mdframed}[backgroundcolor=lightgray!40,topline=false,rightline=false,leftline=false,bottomline=false,innertopmargin=0pt]\begin{mdresult}}{\end{mdresult}\end{mdframed}}

\renewcommand*{\thesection}{\arabic{section}}

\allowdisplaybreaks

\renewcommand{\qed}{\nobreak \ifvmode \relax \else
      \ifdim\lastskip<1.5em \hskip-\lastskip
      \hskip1.5em plus0em minus0.5em \fi \nobreak
      \vrule height0.75em width0.5em depth0.25em\fi}

\usepackage[T1]{fontenc}
\usepackage[utf8]{inputenc}

\newcommand{\rs}{Ruzsa-Szemer\'{e}di\xspace}

\input{macros}

\newcommand{\core}{\textsc{ALG}}
\newcommand{\ignore}[1]{}

\title{Coresets Meet EDCS: Algorithms for Matching and Vertex Cover on Massive Graphs}
\author{Sepehr Assadi\footnote{{\small{\texttt{sassadi@cis.upenn.edu}}}. Supported in part by NSF grant CCF-1617851. Research done in part while the \indent \indent author was a summer intern at Google Research, NYC.}   \\ University of Pennsylvania   \and %\\ University of Pennsylvania  \\ sassadi@cis.upenn.edu \and 
MohammadHossein Bateni\footnote{{\small{\texttt{bateni@google.com}}.}} \\ Google Research  \and 
Aaron Bernstein\footnote{{\small{\texttt{bernstei@gmail.com}}}. Supported in part by Einstein Grant. }  \\ Technical University of Berlin\and
Vahab Mirrokni\footnote{{\small{\texttt{mirrokni@google.com}.}}} \\ Google Research \and%\\ Google Research \\ mirrokni@google.com }%\thanks{\ourinfo{\{sassadi\}@cis.upenn.edu}}}
Cliff Stein\footnote{{\small{\texttt{cliff@ieor.columbia.edu}.}} Research supported in part by NSF grants CCF-1421161 and CCF-1714818.  Some \indent \indent research done while visiting Google.} \\ Columbia University }
\date{}

\begin{document}
\maketitle

\thispagestyle{empty}
\begin{abstract}
\input{abstract}

\end{abstract}
\clearpage
\setcounter{page}{1}

\thispagestyle{empty}
\setcounter{tocdepth}{3}
\tableofcontents
\clearpage
\setcounter{page}{1}

\input{intro}
\input{prelim}

\input{maximum-coreset}

\input{edcs}

\input{coresets}

\input{mpc-algorithms}

\subsection*{Acknowledgements} 
The first author is grateful to his advisor Sanjeev Khanna for the previous collaboration in~\cite{AssadiK17} that was the starting point of this project, to Michael Kapralov 
for helpful discussions regarding the streaming matching problem and results in~\cite{Kapralov13}, and to Krzysztof Onak for helpful discussions regarding the results in~\cite{CzumajLMMOS17}. 

\bibliographystyle{abbrv}
\bibliography{general}

\clearpage
\appendix

\input{app-applications}

\input{app-prelim}

\end{document}

%% file: macros.tex
\newcommand{\toShrink}{-.20cm}
\newcommand{\toShrinkEnu}{-.2cm}

\newenvironment{tbox}{\begin{tcolorbox}[
		enlarge top by=5pt,
		enlarge bottom by=5pt,
	 	breakable,
		 boxsep=0pt,
                  left=4pt,
                  right=4pt,
                  top=10pt,
                  arc=0pt,
                  boxrule=1pt,toprule=1pt,
                  colback=white
                  ]%%
	}
{\end{tcolorbox}}

\newcommand{\textbox}[2]{
{
\begin{tbox}
\textbf{#1}
{#2}
\end{tbox}
}
}

\newcommand{\Leq}[1]{\ensuremath{\underset{\textnormal{#1}}\leq}}
\newcommand{\Geq}[1]{\ensuremath{\underset{\textnormal{#1}}\geq}}

\newcommand{\polylog}{\mbox{\rm  polylog}}

\newcommand{\Ot}{\ensuremath{\widetilde{O}}}
\newcommand{\eps}{\ensuremath{\varepsilon}}
\newcommand{\Paren}[1]{\Big(#1\Big)}
\newcommand{\Bracket}[1]{\Big[#1\Big]}
\newcommand{\bracket}[1]{\left[#1\right]}
\newcommand{\paren}[1]{\ensuremath{\left(#1\right)}\xspace}
\newcommand{\card}[1]{\left\vert{#1}\right\vert}

\newcommand{\set}[1]{\ensuremath{\left\{ #1 \right\}}}

\newcommand{\alg}{\ensuremath{\mathcal{A}}\xspace}

\DeclareMathOperator*{\Exp}{\ensuremath{{\mathbb{E}}}}
\DeclareMathOperator*{\Prob}{\ensuremath{\textnormal{Pr}}}
\renewcommand{\Pr}{\Prob}

\newcommand{\Ex}{\Exp}

\newcommand{\etal}{et al.\xspace}

\newcommand{\MaxMatching}{\ensuremath{\textnormal{\textsf{MaxMatching}}}\xspace}

\newcommand{\Ei}[1]{\ensuremath{E^{(#1)}}}

\newcommand{\Gi}[1]{\ensuremath{G^{(#1)}}}

\renewcommand{\alg}{\ensuremath{\textnormal{\textsf{ALG}}}\xspace}

\newcommand{\Ci}[1]{\ensuremath{C^{(#1)}}}
\newcommand{\Mi}[1]{\ensuremath{M^{(#1)}}}

\newcommand{\EDCS}[1]{\ensuremath{\textnormal{EDCS\ensuremath{(#1)}}\xspace}}
\renewcommand{\deg}[2]{\ensuremath{d_{#1}(#2)}}

\newcommand{\VC}{\ensuremath{\textnormal{\textsf{VC}}}}
\newcommand{\MM}{\ensuremath{\textnormal{\textsf{MM}}}}

\newcommand{\palg}{\ensuremath{\textnormal{\textsf{ParallelAlgorithm}}}\xspace}
\newcommand{\rmatch}{\ensuremath{\textnormal{\textsf{RandomMatch}}}\xspace}

\newcommand{\pedcs}{\ensuremath{\textnormal{\textsf{ParallelEDCS}}}\xspace}

\newcommand{\va}{\ensuremath{V_{\textnormal{\sc alg}}}}
\newcommand{\vh}{\ensuremath{V_{\textnormal{\sc high}}}}
\newcommand{\vm}{\ensuremath{V^-}}
\newcommand{\vr}{\ensuremath{V_{\textnormal{\sc rec}}}}
\newcommand{\vvc}{\ensuremath{V_{\textnormal{\sc vc}}}}

\newcommand{\tDelta}{\widetilde{\Delta}}

\newcommand{\ma}{\ensuremath{M_{\textnormal{\sc alg}}}}
\newcommand{\mh}{\ensuremath{M_{\textnormal{\sc high}}}}
\newcommand{\mr}{\ensuremath{M_{\textnormal{\sc rec}}}}

\newcommand{\cm}{\ensuremath{C^-}}

\newcommand{\Vi}[1]{\ensuremath{V^{(#1)}}}

\newcommand{\hH}{\ensuremath{\widehat{H}}}

\newcommand{\propone}{\textnormal{Property~(P1)}\xspace}
\newcommand{\proptwo}{\textnormal{Property~(P2)}\xspace}

\newcommand{\GEp}{\ensuremath{G^{\textnormal{\textsf{E}}}_p}}
\newcommand{\GVp}{\ensuremath{G^{\textnormal{\textsf{V}}}_p}}

\newcommand{\Esample}{\ensuremath{E_{\textnormal{\textsf{smpl}}}}}
\newcommand{\event}{\ensuremath{\mathcal{E}_{\textsf{valid}}}}

\newcommand{\LL}{\ensuremath{\mathcal{L}}}

%% file: abstract.tex
As massive graphs become more prevalent, there is a rapidly growing
need for  {scalable} algorithms that  solve classical
graph problems, such as maximum matching and minimum vertex cover,  on large datasets.
For massive inputs, several different computational models
have been introduced, including  the streaming model,
the distributed communication model, and the
{massively parallel computation (MPC)} model that is a common
abstraction of MapReduce-style computation.  In each model, algorithms
are analyzed in terms of resources such as space used or rounds
of communication needed, in addition to the more traditional approximation ratio.

\smallskip

In this paper, we give a \emph{single unified approach} that yields better approximation
algorithms for matching and vertex cover in all these models.  The highlights include: 
\begin{itemize}
        \item The first one pass, significantly-better-than-2-approximation for matching in random arrival streams that uses subquadratic space, namely a $(1.5+\eps)$-approximation streaming algorithm that uses $\Ot(n^{1.5})$ space for constant $\eps > 0$. 
		\item The first 2-round, better-than-2-approximation for matching
        in the MPC model that uses subquadratic space per machine, namely a $(1.5+\eps)$-approximation algorithm with $\Ot(\sqrt{mn} + n)$ memory per machine for constant $\eps > 0$. 
\end{itemize}

\smallskip

By building on our unified approach, we further develop parallel algorithms in the MPC model that give a $(1
+ \epsilon)$-approximation to  matching and an $O(1)$-approximation to  vertex cover in only
$O(\log\log{n})$ MPC rounds and $O(n/\polylog{(n)})$ memory per machine.  These 
results settle multiple open questions posed in the recent paper of
Czumaj~\etal [STOC 2018].

\smallskip

We obtain our results by a novel combination of two previously disjoint set of techniques, namely {randomized composable coresets} and {edge degree constrained
subgraphs (EDCS)}. We significantly extend the power of
these techniques and prove several new structural results.  For
example, we show that an EDCS is a sparse certificate for large
matchings and small vertex covers that is quite robust to sampling and
composition.

%% file: intro.tex
\section{Introduction}\label{sec:intro}

As massive graphs become more prevalent, there is a rapidly growing
need for scalable algorithms that  solve classical
graph problems on large datasets.
 When dealing with massive data,
the entire input graph is orders of magnitude larger than the amount of storage
on one processor and hence any algorithm needs to explicitly address
this issue.
For massive inputs, several different computational models
have been introduced, each focusing on certain additional resources needed
to solve large-scale problems.
Some examples include the streaming model,
the distributed communication model, and the
{massively parallel computation (MPC)} model that is a common
abstraction of MapReduce-style computation (see Section~\ref{sec:prelim}
for a definition of MPC). The target resources in these models are the number
of rounds of communication and the local storage
on each machine.

Given the variety of relevant models, there has been a lot of
attention on designing general algorithmic techniques that can
be applicable across a wide range of settings.  We focus on this task
for two prominent graph optimization problems: maximum matching
and minimum vertex cover.  Our main result (Section
\ref{sec:result1}) presents a \textbf{\emph{single unified algorithm}}
that immediately implies significantly improved results (in some or all of the parameters involved) for both
problems in all three models discussed above. For example, in
{random arrival order} streams, our algorithm computes a $(1.5+\eps)$-approximate matching in a single pass with
$\Ot(n^{1.5})$ space; this significantly improves upon the approximation ratio of previous
single-pass algorithms using subquadratic space, and is the first result
to present strong evidence of a separation between random and
adversarial order for matching. Another example is in the MPC model:
Given $\Ot(n^{1.5})$ space per machine, our algorithm computes an
a $(1.5+\eps)$-approximate matching in only 2 MPC rounds; this
significantly improves upon all previous results with a small constant
number of rounds.

Our algorithm is built on the framework of randomized composable
  coreset, which was recently suggested by Assadi and
Khanna~\cite{AssadiK17} as a means to unify different models for
processing massive graphs (see Section \ref{sec:coreset}).  A
common drawback of unified approaches is that although they have the
advantage of versatility, the results they yield are often not as
strong as those that are tailored to one particular model. It is
therefore perhaps surprising that we can design essentially a single
algorithm that improves upon the state-of-the-art algorithms in all
three models discussed above simultaneously. Our approach for the
matching problem notably goes significantly beyond a
$2$-approximation, which is a notorious barrier for matching in all
the models discussed above.

We also build on our techniques to achieve a second result (Section
\ref{sec:result2}) particular to the MPC model. We show that
when each machine has only $O(n)$ space (or even $O(n/\polylog(n)))$,
$O(\log\log{n})$ rounds suffice to compute a $(1+\eps)$-approximate
matching or a $O(1)$-approximate vertex cover. This improves
significantly upon the recent breakthrough of Czumaj et
al. \cite{CzumajLMMOS17}, which does not extend to vertex cover, and
requires $O(\log\log^2(n))$ rounds. Our results in this part settle
multiple open questions posed by Czumaj~\etal~\cite{CzumajLMMOS17}.

\subsection{Randomized Composable Coresets}
\label{sec:coreset}

Two examples of general techniques widely used for processing massive data sets are 
linear sketches (see e.g. ~\cite{AhnGM12Linear,AGM12,KapralovLMMS14,AssadiKLY16,ChitnisCEHMMV16,BulteauFKP16,KW14,BhattacharyaHNT15,McGregorTVV15}) and composable coresets
(see e.g.~\cite{BadanidiyuruMKK14,BalcanEL13,BateniBLM14,IndykMMM14,MirzasoleimanKSK13,MirrokniZ15,AssadiK17}).
Both proceed by arbitrarily partitioning the data into smaller pieces,
computing a small-size summary of each piece, and then showing that these summaries can be combined
into a small-size summary of the original data set. This approach has a wide range of applications,
but strong impossibility results are known for both techniques for the two problems of {maximum matching} and {minimum vertex cover} that we study in this paper~\cite{AssadiKLY16}.

Recently, Assadi and Khanna~\cite{AssadiK17} turned to the notion of randomized composable coresets---originally introduced in the context of submodular maximization by Mirrokni and Zadimoghadam~\cite{MirrokniZ15} (see also~\cite{BarbosaENW15})---to bypass these strong impossibility results. 
The idea is to partition the graph into \emph{random} pieces rather than arbitrary ones. 
The authors in \cite{AssadiK17}
designed randomized composable coresets for matching and vertex cover, but although
this led to unified algorithms for many models of computation, the resulting bounds were
still for the most part weaker than the state-of-the-art algorithms tailored to each particular model. 

We now define randomized composable coresets in more detail; for brevity, we refer to them as randomized coresets.
Given a graph $G(V,E)$, with $m=|E|$ and $n=|V|$,  consider a random partition of $E$ into $k$ edge sets
$\set{\Ei{1},\ldots,\Ei{k}}$; each edge in $E$ is sent to exactly one of the $\Ei{i}$,
picked uniformly at random, thereby partitioning graph $G$ into $k$ subgraphs 
$\Gi{i}(V, \Ei{i})$.

\begin{definition}[Randomized Composable Coreset~\cite{MirrokniZ15,AssadiK17}]
\label{def:randomized-coreset}
Consider an algorithm $\core$ that takes as input an arbitrary graph and returns a 
\textbf{subgraph} $\core(G) \subseteq G$. $\core$ is said to output an $\alpha$-approximate randomized coreset for maximum matching if given any graph $G(V,E)$ and a random $k$-partition of $G$ into $\Gi{i}(V,\Ei{i})$, the size of the maximum matching
in $\core(\Gi{1}) \cup \ldots \cup \core(\Gi{k})$ is an $\alpha$-approximation to the size of the maximum matching in $G$ with high probability.
We refer to the \textbf{number of edges} in the returned subgraph by $\core$ as the \textbf{size} of the coreset. 
Randomized coresets are defined analogously for minimum vertex cover and other graph problems. 
\end{definition}

It is proven in~\cite{AssadiK17} that any $O(1)$-approximate randomized coreset for matching or vertex cover has size $\Omega(n)$. Thus, similarly to~\cite{AssadiK17}, we focus
on designing randomized coresets of size $\Ot(n)$, which is optimal within logarithmic factors. The following proposition states some immediate applications of randomized coresets. 

\begin{proposition}\label{prop:coreset-application}
Suppose $\core$ outputs an $\alpha$-approximate randomized coreset of size $\Ot(n)$ for problem $P$ (e.g. matching).
Let $G(V,E)$ be a graph with $m = \card{E}$ edges. This yields:
	\begin{enumerate}
		\item A parallel algorithm in the MPC model that with high probability outputs an $\alpha$-approximation to $P$ in two rounds with $\Ot(\sqrt{m/n})$ machines, each with $\Ot(\sqrt{mn} + n) = \Ot(n^{1.5}$) memory. 
		\item A streaming algorithm that on \emph{random arrival} streams outputs an $\alpha$-approximation to $P(G)$ with high probability using $\Ot(\sqrt{mn} + n) = \Ot(n^{1.5}$) space. 
		\item A simultaneous communication protocol that on \emph{randomly partitioned} inputs computes an $\alpha$-approximation to $P(G)$ with high probability using $\Ot(n)$ communication per machine/player. 
	\end{enumerate}
\end{proposition}
\noindent
A proof of Proposition~\ref{prop:coreset-application} can be found in Appendix~\ref{app:applications}.

%%Any $\alpha$-approximate randomized
%%coreset of size $\Ot(n)$ for matching or vertex cover immediately yields the following algorithms for these problems on a graph $G(V,E)$ with $n$ vertices and $m$ edges: (see Proposition~\ref{prop:coreset-application} for more details and proofs.)  
%%
%%\begin{itemize}
%%	\item \textbf{Streaming:} An $\alpha$-approximation single-pass streaming algorithm on \emph{random arrival} streams using $\Ot(\sqrt{mn} + n) = \Ot(n^{1.5}$) space. 
%%	\item \textbf{MPC:} An $\alpha$-approximation MPC algorithm in two rounds with $O(\sqrt{m/n})$ machines, each with $\Ot(\sqrt{mn} + n) = \Ot(n^{1.5}$) memory. 
%%	\item \textbf{Distributed:} An $\alpha$-approximation simultaneous communication protocol on \emph{randomly partitioned} inputs using $\Ot(n)$ communication per machine/player. 
%%\end{itemize}
%%
%%%\noindent See Proposition~\ref{prop:coreset-application} for more details and proofs.
%%%%
%%%%

\subsection{First Result: Improved Algorithms via a New Randomized Coreset}
\label{sec:result1}

We start by studying the previous randomized coreset of~\cite{AssadiK17} for matching, which was simply to pick a maximum matching of each
machine's subgraph as its coreset. This is arguably the most natural approach to the problem and results in truly sparse subgraphs
(maximum degree one). As a warm-up to our main results, we present a simpler and improved analysis (compared to that in~\cite{AssadiK17}), which shows that this coreset
achieves a $3$-approximation (vs. $9$-approximation proven in~\cite{AssadiK17}). We also show that there exist graphs on which the approximation ratio of this coreset
is at least $2$. This suggests that to achieve a better than $2$ approximation, fundamentally different ideas are needed
which brings us to our first main result. 
\smallskip

\begin{result}\label{res:coresets}
	There exist randomized composable coresets of size $\Ot(n)$ that for any constant $\eps > 0$, give a $(3/2+\eps)$-approximation for maximum matching and $(2+\eps)$-approximation for minimum vertex cover with high probability. 
\end{result}

\noindent Our results improve upon the randomized coresets
of~\cite{AssadiK17} that obtained $O(1)$ and $O(\log{n})$
approximation to matching and vertex cover, respectively.  We notably
go beyond the ubiquitous 2-approximation barrier for matching 
(in Section~\ref{sec:maxmatching-lower}, we show that the previous approach 
of ~\cite{AssadiK17} provably cannot go below $2$).  {Result~\ref{res:coresets}
  yields a unified framework that improves upon the state-of-the art
  algorithms for matching and vertex cover across several
  computational models in one or all parameters involved.}

\paragraph{\underline{First implication: streaming.}} 

We consider {single-pass} streaming algorithms. Computing a $2$-approximation for matching (and vertex cover) in $O(n)$ space is
trivial: simply maintain a maximal matching. Going beyond this barrier has 
remained one of the central open questions in the graph streaming literature since the introduction of the field~\cite{FKMSZ05}.
No $o(n^2)$-space algorithm is known for this task on {adversarially ordered} streams and the 
lower bound result by Kapralov~\cite{Kapralov13} (see also~\cite{GoelKK12}) proves that an $\paren{\frac{e}{e-1}}$-approximation requires $n^{1+\Omega(1/\log\log{n})}$ space.
To make progress on this fascinating open question, Konrad~\etal~\cite{KonradMM12} suggested the study of matching in {random arrival} streams. They presented 
an algorithm with approximation ratio strictly better than $2$, namely $2-\delta$ for $\delta \approx 0.002$, in $O(n)$ space over random streams. 
A direct application of our Result~\ref{res:coresets} improves the approximation ratio of this algorithm significantly albeit at a cost of a larger space requirement. 
\begin{corollary}\label{cor:coreset-stream}
	There exists a single-pass streaming algorithm on \emph{random arrival streams} that uses $\Ot(n^{1.5})$ space and with high probability (over the randomness of the stream) achieves a $(3/2+\eps)$-approximation
	to the maximum matching problem for constant $\eps > 0$. 
\end{corollary}

Our results provide the first strong evidence of a separation between random-order and
adversarial-order streams for matching, as it is the first algorithm that beats 
the ratio of $\paren{\frac{e}{e-1}}$, which is known to be ``hard'' on adversarial streams ~\cite{Kapralov13}.
Although the lower bound of ~\cite{Kapralov13} does not preclude 
achieving the bounds of Corollary \ref{cor:coreset-stream} in an adversarial
order (because our space is $\Ot(n^{1.5})$ rather than $\Ot(n)$),
the proof in~\cite{Kapralov13} (see
also~\cite{GoelKK12}) suggests that achieving such bounds is
ultimately connected to further understanding of \rs
graphs, a notoriously hard problem in additive combinatorics 
(see e.g. \cite{Gowers01, FoxHS15,AlonMS12}).
From a different perspective, 
most (but not all) streaming lower bounds are proven by bounding
the (per-player) communication complexity of the problem in the {blackboard} communication model,
including the $\paren{\frac{e}{e-1}}$ lower bound of~\cite{Kapralov13}.
Our algorithm in Result~\ref{res:coresets} can be implemented with $\Ot(n)$ (per-player) communication in this model which goes strictly
below the lower bound of~\cite{Kapralov13}, thus establishing the first \emph{provable separation} between
adversarial- and random-partitioned inputs in the blackboard communication model for computing a matching.

\paragraph{\underline{Second implication: MPC.}} Maximum matching and minimum vertex cover are among the most studied graph optimization problems in the MPC and other MapReduce-style computation models
~\cite{AhnGM12Linear,LattanziMSV11,AhnG15,AssadiK17,CzumajLMMOS17,BehnezhadDETY17,HarveyLL18}. As an application of Result~\ref{res:coresets}, 
we obtain efficient MPC algorithms for matching and vertex cover in only {two} rounds of computation.

\begin{corollary}\label{cor:coreset-mpc}
	There exist MPC algorithms that with high probability achieve a $(3/2+\eps)$-approximation to matching and a 
        $(2+\eps)$-approximation to vertex cover in {two} MPC rounds and
         $\Ot(\sqrt{mn} + n)$ memory per machine for constant $\eps > 0$\footnote{The approximation
          factor for vertex cover degrades to $4$ if one requires
          local computation on each machine to be polynomial time; see
          Remarks~\ref{rem:computational-time}
          and~\ref{rem:vc-polytime}.}.
\end{corollary}

\noindent It follows from the results of \cite{AssadiKLY16} that sub-quadratic memory is not possible with one MPC round,
so two rounds is {optimal}. Furthermore, our implementation only requires one round if the
input is distributed randomly in the first place; see~\cite{MirrokniZ15} for details on when this assumption applies.

Our algorithms outperform the previous algorithms of~\cite{AssadiK17} for matching and vertex cover in terms of approximation ratio ($3/2$ vs.~$O(1)$ and $2$ vs.~$O(\log{n})$), while memory and round complexity are the same. Our matching algorithm outperforms the $2$-approximate maximum matching
algorithm of Lattanzi~\etal~\cite{LattanziMSV11} in terms of both the approximation ratio ($3/2$ vs. $2$) and round complexity ($2$ vs. $6$) within the same memory. Our result for the matching problem is particularly interesting 
as all other MPC algorithms~\cite{AhnGM12Linear,AhnG15,BehnezhadDETY17} that can achieve a better than two approximation (which is also a natural barrier for matching algorithms across different models) require a large (unspecified) constant
number of rounds. Achieving the optimal 2 rounds is significant in this context, since the round complexity of MPC algorithms determines the dominant cost of the computation (see, e.g.~\cite{LattanziMSV11,BeameKS13}), and hence 
minimizing the number of rounds is the primary goal in this model.

\paragraph{\underline{Third implication: distributed simultaneous communication.}} Maximum matching (and to a lesser degree vertex cover) has been studied previously in the 
simultaneous communication model owing to many applications of this model, including in achieving round-optimal distributed algorithms~\cite{AssadiK17}, proving lower bounds for dynamic graph streams~\cite{Konrad15,AiHLW16,AssadiKLY16,AssadiKL17}, 
and applications to mechanism design~\cite{DobzinskiNO14,AlonNRW15,Dobzinski16}. As an application of Result~\ref{res:coresets}, we obtain the following corollary. 
\begin{corollary}\label{cor:coreset-sim}
	There exist simultaneous communication protocols on \emph{randomly partitioned inputs} that achieve $(3/2+\eps)$-approximation to matching and $(2+\eps)$-approximation to vertex cover with high probability
	(over the randomness of the input partitioning) with only $\Ot(n)$ communication per machine/player for constant $\eps > 0$. 
\end{corollary}

\noindent This result improves upon the $O(1)$ and $O(\log{n})$ approximation  of~\cite{AssadiK17} (on randomly partitioned inputs) for matching and vertex cover that were also designed by using randomized coresets. Our protocols 
achieve optimal communication complexity (up to $\polylog{(n)}$ factors)~\cite{AssadiK17}.

\subsection{Second Result: MPC with Low Space Per Machine}\label{sec:result2}

Our second result concerns the MPC model with per-machine memory $O(n)$ or even $O(n/\polylog{(n)})$. This is achieved by 
extending our Result~\ref{res:coresets} from random edge-partitioned subgraphs (as in randomized coresets) to random vertex-partitioned subgraphs (which we explain further below).

\smallskip

\begin{result}\label{res:mpc}
	There exists an MPC algorithm that for any constant $\eps > 0$, with high probability, gives a $(1+\eps)$-approximation to maximum matching and $O(1)$-approximation to 
	minimum vertex cover in $O(\log\log{n})$ MPC rounds using only $O(n /\polylog{(n)})$ memory per 
	machine. 
\end{result}

Given an existing black-box reduction~\cite{LotkerPP15} (see also~\cite{CzumajLMMOS17}), our Result~\ref{res:mpc} immediately implies 
a $(2+\eps)$-approximation algorithm for maximum {weighted} matching in 
the same $O(\log\log(n))$ rounds, though with the memory per machine increased to $O(n\log(n))$.

Prior to~\cite{CzumajLMMOS17}, all MPC
algorithms for matching and vertex cover~\cite{LattanziMSV11,AhnGM12Linear,AhnG15} 
required $\Omega\paren{\frac{\log{n}}{\log\log{n}}}$ rounds to achieve
$O(1)$ approximation when the memory per machine was restricted to
$\Ot(n)$ (which is arguably the most natural choice of parameter,
similar-in-spirit to the semi-streaming
restriction~\cite{FKMSZ05,M14}). 
Recently, Czumaj~\etal~\cite{CzumajLMMOS17} presented an (almost) $2$-approximation algorithm for maximum matching that requires $O(n)$ (even $n/(\log{n})^{O(\log\log{n})}$) memory per machine
and only $O\paren{(\log\log{n})^2}$ MPC rounds. 
Result~\ref{res:mpc} improves upon this result on several fronts: $(i)$ we improve the round complexity of the matching algorithm
to $O(\log\log{n})$, resolving a conjecture of~\cite{CzumajLMMOS17} in the affirmative, $(ii)$ we obtain an $O(1)$ approximation to vertex cover, answering another open question of~\cite{CzumajLMMOS17}, and
$(iii)$ we achieve all these using a considerably simpler algorithm and analysis than~\cite{CzumajLMMOS17}.  

\paragraph{Comparison to results published after the appearance of our paper.}
After an earlier version of our paper was shared on arXiv \cite{AssadiBBMSarxiv2017}, Ghaffari et al.~\cite{GhaffariGMR18} presented a result very similar to 
our Result~\ref{res:mpc}: their bounds are exactly the same for matching, while for vertex cover they achieve a better
approximation in  the same asymptotic number of rounds: $(2+\eps)$-approximation vs. our $O(1)$ approximation. Techniques-wise, our approaches are entirely different: 
the algorithms in ~\cite{GhaffariGMR18} are based on an earlier round-compression technique of~\cite{CzumajLMMOS17},
and require an intricate local algorithm
and analysis to ensure consistency between machines;
see Section \ref{sec:techniques} below for more details.

Note that only our Result~\ref{res:mpc} is shared with the later paper of Ghaffari et al. \cite{GhaffariGMR18}: Result~\ref{res:coresets} appears only in our paper, and is entirely specific to the particular techniques that we use.

\subsection{Our Techniques}\label{sec:techniques}

Both of our results are based on a novel
application of \emph{edge degree constrained subgraphs (EDCS)} that
were previously introduced by Bernstein and Stein~\cite{BernsteinS15}
for maintaining large matchings in dynamic graphs. 
Previous work
on EDCS~\cite{BernsteinS15,BernsteinS16} focused on how large a matching
an EDCS contains and how it can be maintained in a dynamic
graph. For the two results of this paper, we instead focus on the structural properties of the EDCS,
and prove several new facts in this regard.

For Result~\ref{res:coresets}, we identify the EDCS as a sparse
certificate for large matchings and small vertex covers
which are quite robust to sampling and composition: an ideal
combination for a randomized coreset.
For Result~\ref{res:mpc}, we use the following recursive procedure, which
crucially relies upon on the robustness properties of the EDCS proved in 
Result~\ref{res:coresets}: we repeatedly compute an EDCS of the
underlying graph in a distributed fashion, redistribute it again 
amongst multiple machines, and recursively solve the problem on
this EDCS to compute an $O(1)$-approximation to matching and vertex
cover. We therefore limit the memory on each machine to only
$O(n)$ (even $O(n/\polylog{(n)})$) at the cost of increasing the number of rounds from $O(1)$ to
$O(\log\log{n})$.  Additional ideas are needed to ensure that the
approximation ratio of the algorithm does not increase beyond a fixed
constant as a result of repeatedly computing an EDCS of the current
graph in $O(\log\log{n})$ iterations.

%We obtain our coresets in Result~\ref{short-res:coresets} using a novel
%application of \emph{edge degree constrained subgraphs (EDCS)} that
%were previously introduced by Bernstein and Stein~\cite{BernsteinS15}
%for maintaining large matchings in dynamic graphs. While previous work
%on EDCS~\cite{BernsteinS15,BernsteinS16} focused on how large a matching
%an EDCS contains and how it can be maintained efficiently in a dynamic
%graph, in this paper we study several new structural
%properties of the EDCS itself. Our results identify these subgraphs as
%sparse certificates for large matchings and small vertex covers
%which are quite robust to sampling and composition, an ideal
%combination for a randomized  coreset.

\paragraph{Comparison of techniques.} 
Result~\ref{res:coresets} uses the definition of EDCS from \cite{BernsteinS15,BernsteinS16} but uses it in an entirely different setting, and hence we prove and use novel properties of EDCS in this work.

Result~\ref{res:mpc} relies on the high-level technique of
vertex sampling from Czumaj~\etal~\cite{CzumajLMMOS17}: instead
of partitioning the edges of the graph, each machine receives a
random sample of the vertices, and works on the resulting
induced subgraph.  Other than this starting point, our approach
proceeds along entirely different lines from~\cite{CzumajLMMOS17}, in
terms of both the local algorithm computed on each subgraph and in the
analysis.  The main approach in~\cite{CzumajLMMOS17} is \emph{round
  compression}, which corresponds to compressing multiple rounds of a
particular distributed algorithm into smaller number of MPC rounds by
maintaining a consistent state across the local algorithms computed on
each subgraph (using a highly non-trivial local algorithm and
analysis). Our results, on the other hand, do \emph{not} correspond to
a round compression approach at all and we do not require any
consistency in the local algorithm on each machine. Instead, we rely
on structural properties of the EDCS that we prove in this
paper, \emph{independent of the algorithms that compute these
  subgraphs}. This allows us to bypass many of the technical
difficulties arising in maintaining a consistent state across
different machines which in turn results in improved bounds and a
considerably simpler algorithm and analysis.

\subsection{Related Work}\label{sec:related}
Maximum matching and minimum vertex cover are among the most studied problems in the context of massive graphs including in MPC model and MapReduce-style
computation~\cite{AhnGM12Linear,LattanziMSV11,AhnG15,AssadiK17,CzumajLMMOS17,BehnezhadDETY17,HarveyLL18,GhaffariGMR18}, streaming 
algorithms~{\cite{McGregor05,FKMSZ05,EKS09,EpsteinLMS11,AhnGM12Linear,GoelKK12,KonradMM12,AGM12,AG13,GO13,Kapralov13,KapralovKS14,CS14,ChitnisCHM15,M14,AhnG15,%
EsfandiariHLMO15,Konrad15,AssadiKLY16,ChitnisCEHMMV16,McGregorV16,EsfandiariHM16,AssadiKL17,PazS17}, simultaneous communication model and similar distributed 
models~\cite{GO13,DobzinskiNO14,HuangRVZ15,AlonNRW15,AssadiKLY16,AssadiKL17,AssadiK17}, dynamic graphs~\cite{NeimanS13,Solomon16,BaswanaGS15,OnakR10,BernsteinS15,BernsteinS16}, and 
sub-linear time algorithms~\cite{ParnasR07,HassidimKNO09,NguyenO08,OnakRRR12,YoshidaYI12}. 
Beside the results mentioned already, most relevant to our work are 
the $\polylog{(n)}$-space $\polylog{(n)}$-approximation algorithm of~\cite{KapralovKS14} for estimating the \emph{size} of a maximum matching in random stream, and the $(3/2)$-approximation
communication protocol of~\cite{GoelKK12} when the input is (adversarially) partitioned between two parties and the communication is from one party to the other one (as opposed to simultaneous which we studied). 
However, the techniques in these results and ours are completely disjoint. %We also remark that random partitioning of the input in communication models have been considered previously, for example in~\cite{ChakrabartiCM08}. 

Coresets, composable coresets, and randomized composable coresets are respectively introduced in~\cite{AgarwalHV04},~\cite{IndykMMM14}, and~\cite{MirrokniZ15}. Composable coresets 
have been studied previously in nearest neighbor search~\cite{AbbarAIMV13}, diversity maximization~\cite{IndykMMM14,ZadehGMZ17}, clustering~\cite{BalcanEL13,BateniBLM14}, and
submodular maximization~\cite{IndykMMM14,MirrokniZ15,BadanidiyuruMKK14,BarbosaENW15,BarbosaENW16}. Moreover, while not particularly termed a composable coreset, the ``merge and reduce'' technique
in graph streaming literature (see~\cite{M14}, Section~2.2) is identical to composable coresets. 

%% file: prelim.tex
\section{Preliminaries}\label{sec:prelim}

\paragraph{Notation.} For a graph $G(V,E)$, we use $\MM(G)$ to denote the maximum matching size in $G$ and $\VC(G)$ to denote the minimum vertex cover size. 
For any subset of vertices $V' \subseteq V$ and any subset of edges $E' \subseteq E$, we use $V'(E')$ to denote the set of vertices in $V'$ that are incident on edges of $E'$ and 
$E'(V')$ to denote the set of edges in $E'$ that are incident on vertices of $V'$. For any vertex $v \in V$, we use $\deg{G}{v}$ to denote the degree of $v$ in the graph $G$. 
%For integers $a,b,c$, we use $a = b \pm c$ to denote that $a \in [b-c,b+c]$. 

We use capital letters to denote random variables.  Let $\set{X_i}_{i=1}^{n}$ and $\set{Y_i}_{i=1}^{n}$ be a sequence of random variables on a common probability space such that $\Ex\bracket{X_i \mid Y_1,\ldots,Y_{i-1}} = X_{i-1}$ for all $i$. 
The sequence $\set{X_i}$ is referred to as a {martingale} with respect to $\set{Y_i}$. A summary of concentration bounds we use in this paper appears in Appendix~\ref{app:concentrations}.

\paragraph{Sampled Subgraphs.} Throughout the paper, we work with two different notion of sampling a graph $G(V,E)$. For a parameter $p \in (0,1)$, 
\begin{itemize}
	\item A graph $\GEp(V,E_p)$ is an \emph{edge sampled subgraph} of $G$ iff the vertex set of $\GEp$ and $G$ are the same and every edge in $E$ is picked independently and with probability $p$ in $E_p$. 
	\item A graph $\GVp(V_p,E_p)$ is a \emph{vertex sampled (induced) subgraph} of $G$ iff every vertex in $V$ is sampled in $V_p$ independently and with probability $p$ and $\GVp$ is the induced subgraph of $G$ on $V_p$
\end{itemize}
%%
%%we say that a graph $\GEp(V,E_p)$ is an 
%%\emph{edge sampled subgraph} of $G$ iff the vertex set of $\GEp$ and $G$ are the same and every edge in $E$ is picked independently and with probability $p$ in $E_p$ (and hence in $\GEp$). We say that a graph 
%%$G^V_p(V_p,E_p)$ is a \emph{vertex sampled (induced) subgraph} of $G$ iff every vertex in $V$ is sampled in $V_p$ (and hence in $G^V_p$) independently and with probability $p$ and $G^V_p$ is the induced 
%%subgraph of $G$ on vertices $V_p$. 

\subsection{The Massively Parallel Computation (MPC) Model}\label{sec:mpc} 

We adopt the most stringent model of modern parallel computation among~\cite{KarloffSV10,GoodrichSZ11,AndoniNOY14,BeameKS13}, the so-called \emph{Massively Parallel Computation (MPC)} model of~\cite{BeameKS13}. 
Let $G(V,E)$ with $n:= \card{V}$ and $m := \card{E}$ be the input graph. 
In this model, there are $p$ machines, each with a memory of size $s$ and one typically requires that both $p,s = m^{1-\Omega(1)}$ i.e., polynomially smaller than the input size~\cite{KarloffSV10,AndoniNOY14}. 
% such that $p \cdot s = \Ot(m)$. 
Computation proceeds in synchronous rounds: in each round, each machine performs some local computation
and at the end of the round machines exchange messages to guide the computation for the next round. All messages sent and received by each machine in each round have to fit into the local memory of the machine. This in particular means that
the length of the messages on each machine is bounded by $s$ in each round. At the end, the machines collectively output the solution.  %Following~\cite{KarloffSV10,AndoniNOY14}, we restrict both $k,s = N^{1-\Omega(1)}$. 

\subsection{Basic Graph Theory Facts} \label{sec:graph-facts}

\begin{fact}\label{fact:vc-matching}
	For any graph $G(V,E)$, $\MM(G) \leq \VC(G) \leq 2 \cdot \MM(G)$. 
\end{fact}

The  following propositions are well-known. % and are proven in Appendix~\ref{app:graph-facts} for completeness. 
\begin{proposition}\label{prop:vc-matching-alpha}
	Suppose $M$ and $V'$ are respectively, a matching and a vertex cover of a graph $G$ such that $\alpha \cdot \card{M} \geq \card{V'}$; then, both $M$ and $V'$ are $\alpha$-approximation 
	to their respective problems. 
\end{proposition}
\begin{proof}
	$\VC(G) \Geq{Fact~\ref{fact:vc-matching}} \MM(G) \geq \card{M} \geq \frac{1}{\alpha} \cdot \card{V'} \geq \frac{1}{\alpha} \cdot \VC(G) \Geq{Fact~\ref{fact:vc-matching}} \frac{1}{\alpha} \cdot \MM(G)$. 
\end{proof}

\begin{proposition}\label{prop:matching-high-degree}
	Suppose $G(V,E)$ is a graph with maximum degree $\Delta$ and $\vh$ is the set of all vertices with degree at least $\gamma \cdot \Delta$ in $G$ for $\gamma \in (0,1)$. Then,  $\MM(G) \geq \frac{\gamma \cdot \Delta}{2 \cdot (\Delta+1)} \cdot \card{\vh}$. 
\end{proposition}
\begin{proof}
	By Vizing's theorem~\cite{Vizing64}, $G$ can be edge colored by at most $\Delta + 1$ colors. As each color class forms a matching, this means that there exists a matching $M$ in $G$ with $\card{M} \geq \frac{\card{E}}{\Delta + 1}$. Moreover, 
	we have $\card{E} \geq \frac{1}{2} \cdot \card{\vh} \cdot \gamma \cdot \Delta$, finalizing the proof. 
\end{proof}

\subsection{Edge Degree Constrained Subgraph (EDCS)} \label{sec:edcs-def}

We introduce edge degree constrained subgraphs (EDCS) in this section and present several of their properties which are proven in previous work. {We emphasize that all other properties of EDCS proven in the subsequent sections 
are new to this paper.} 

An EDCS is defined formally as follows. 

\begin{definition}[\!\!\cite{BernsteinS15}]\label{def:edcs}
	For any graph $G(V,E)$ and integers $\beta \geq \beta^- \geq 0$, an \emph{edge degree constraint subgraph (EDCS)} $(G,\beta,\beta^-)$ is a subgraph $H:=(V,E_H)$ of $G$ with the following two properties: 
	\begin{enumerate}[(P1)]
		\item\label{long-property1} For any edge $(u,v) \in E_H$: $\deg{H}{u} + \deg{H}{v} \leq \beta$. 
		\item\label{long-property2} For any edge $(u,v) \in E \setminus E_H$: $\deg{H}{u} + \deg{H}{v} \geq \beta^-$.  
	\end{enumerate}
	We sometimes abuse the notation and use $H$ and $E_H$ interchangeably. 
\end{definition}

In the remainder of the paper, we use the terms ``\propone'' and ``\proptwo'' of EDCS to refer to the first and second items in Definition~\ref{def:edcs} above.

One can prove the existence of an $\EDCS{G,\beta,\beta^-}$ for any graph $G$ and parameters $\beta^- < \beta$ using the results in~\cite{BernsteinS16} (Theorem 3.2) which in fact shows how to maintain
an EDCS efficiently in the dynamic graph setting. As we are only interested in existence of EDCS in this paper, we provide a simpler and self-contained proof of this fact in Appendix~\ref{app:edcs-exists}, 
which also implies a simple polynomial time algorithm for computing any EDCS of a given graph $G$.  

%for any choices of $\beta^- < \beta$ and any graph $G(V,E)$, an $\EDCS{G,\beta,\beta^-}$ indeed exists. 
\begin{lemma}\label{lem:edcs-exists}
	Any graph $G$ contains an $\EDCS{G,\beta,\beta^-}$ for any parameters $\beta > \beta^-$. 
\end{lemma}

It was shown in~\cite{BernsteinS15} (bipartite graphs) and~\cite{BernsteinS16} (general graphs) that 
for appropriate parameters and EDCS always 
contains an (almost) $3/2$-approximate maximum matching of $G$. 
Formally: 

\begin{lemma}[\!\cite{BernsteinS15,BernsteinS16}]\label{lem:edcs-matching}
	Let $G(V,E)$ be any graph and $\eps < 1/2$ be a parameter. For parameters $\lambda \geq \frac{\eps}{100}$, $\beta \geq 32\lambda^{-3}$, and $\beta^- \geq (1-\lambda) \cdot \beta$, in any subgraph $H:= \EDCS{G,\beta,\beta^-}$, 
	$\MM(G) \leq \paren{\frac{3}{2} + \eps} \cdot \MM(H)$. 
\end{lemma}

Lemma~\ref{lem:edcs-matching} implies that an EDCS of a graph $G(V,E)$ preserves the maximum matching of $G$ approximately. We also show a similar result for vertex cover. The basic idea is that in addition to computing a vertex cover for the subgraph $H$(to cover all the edges in $H$), we also add to the vertex cover all vertices that
have degree at least $\geq \beta^-/2$ in $H$, which by \proptwo of an EDCS covers all edges
in $G \setminus H$. %The proof is deferred to Appendix~\ref{app:edcs-vc}. 

\begin{lemma}\label{lem:edcs-vc}
	Let $G(V,E)$ be any graph, $\eps < 1/2$ be a parameter, and $H:= \EDCS{G,\beta,\beta^-}$ for parameters $\beta \geq \frac{4}{\eps}$ and $\beta^- \geq \beta \cdot \paren{1-\eps/4}$. 
	Suppose $\vh$ is the set of vertices $v \in V$ with $\deg{H}{v} \geq \beta^-/2$ and $\vvc$ is a minimum vertex cover of $H$; then $\vh \cup \vvc$ is a vertex cover of $G$ with size at most $(2+\eps) \cdot \VC(H)$ (note that $\VC(H) \leq \VC(G)$). 
\end{lemma}
\begin{proof}
	We first argue that $\vh \cup \vvc$ is indeed a feasible vertex cover of $G$. To see this, notice that any edge $e \in H$ is covered by $\vvc$, and moreover by $\proptwo$ of EDCS, any edge $e \in E \setminus H$ has at least one endpoint with 
	degree at least $\beta^-/2$ in $H$ and hence is covered by $\vh$. In the following, we bound the size of $\vh \setminus \vvc$ by $(1+\eps)\cdot\card{\vvc}$, which finalizes the proof as clearly ${\card{\vvc}} = {\VC(H)}$. 
	
	Define $S := \vh \setminus \vvc$ and let $N(S)$ be the set of all neighbors of $S$ in the EDCS $H$. Since $S$ is not part of the vertex cover $\vvc$ of $H$, we should have $N(S) \subseteq \vvc$ as otherwise some edges between $S$ and $N(S)$
	would not be covered by the vertex cover $\vvc$. Now, since any vertex in $S$ has degree at least $\beta^-/2$, we should have that degree of any vertex in $N(S)$ is at most $\beta-\beta^-/2$ in order to satisfy $\propone$ of EDCS $H$. 
	Let $E(S)$ denote the set of edges incident on $S$ in $H$. As all vertices in $S$ belong to $\vh$, we have that $\card{E(S)} \geq \card{S} \cdot \beta^-/2$. On the other hand, 
	as all edges incident on $S$ are going into $N(S)$ by definition, and since degree of vertices in $N(S)$ are bounded by $\beta-\beta^-/2$, we have $\card{E(S)} \leq \card{N(S)} \cdot (\beta-\beta^-/2)$. As such, 
	\begin{align*}
		\card{S} \cdot \beta^-/2 \leq \card{N(S)} \cdot \paren{\beta-\beta^-/2} \leq \card{\vvc} \cdot (1+\eps) \cdot \beta^-/2,
	\end{align*}
	implying that $\card{S} \leq (1+\eps) \cdot \card{\vvc}$, finalizing the proof. 
\end{proof}

%% file: maximum-coreset.tex
\newcommand{\Mstar}{\ensuremath{M^{\star}}}

\newcommand{\Gm}{\ensuremath{G^{-}}}
\newcommand{\Em}{\ensuremath{E^{-}}}

\section{Warmup: A $3$-Approximation Coreset for Matching}\label{sec:maxmatching}

A natural randomized coreset for the matching problem was previously proposed by~\cite{AssadiK17}: simply compute
a {maximum matching} of each graph $\Gi{i}$. We refer to this randomized coreset as the \MaxMatching coreset. 
It was shown in~\cite{AssadiK17} that \MaxMatching is an $O(1)$-approximation randomized coreset for the 
matching problem (the hidden constant in the O-notation was bounded by $9$ in~\cite{AssadiK17}). As a warm up, we propose a better analysis of this randomized coreset in this section. 

\begin{theorem}\label{thm:maxmatching-3}
	Let $G(V,E)$ be a graph with $\MM(G) = \omega(k\log{n})$ and $\Gi{1},\ldots,\Gi{k}$ be a random $k$-partition of $G$. 
	Any maximum matching of the graph $\Gi{i}$ is a $\paren{3+o(1)}$-approximation randomized composable coreset of size $O(n)$ for the maximum matching problem. 
\end{theorem}

\paragraph{Assumption on $\MM(G)$.} In this section, we follow~\cite{AssadiK17} in assuming that $\MM(G) = \omega(k\log{n})$ since otherwise we can immediately obtain a (non-randomized) composable
coreset with approximation ratio one (an exact maximum matching) and size $\Ot(k^2)$ for the matching problem using the results in~\cite{ChitnisCEHMMV16}.

A crucial building block in our proof of Theorem~\ref{thm:maxmatching-3} is a new concentration result for the size of maximum matching in edge sampled subgraphs that we prove in the next section. This result is quite general
and can be of independent interest. 

\subsection{Concentration of Maximum Matching Size under Edge Sampling}\label{sec:concentration}

Let $G(V,E)$ be any arbitrary graph and $p \in (0,1)$ be a parameter (possibly depending on size of the graph $G$). Define $\GEp(V,E_p)$ as a subgraph of $G$ obtained by sampling each edge in $E$ independently
and with probability $p$, i.e., an edge sampled subgraph of $G$. 
We show that $\MM(G_p)$ is concentrated around its expected value. 

\begin{lemma}\label{lem:concentration}
	Let $G(V,E)$ be any arbitrary graph, $p \in (0,1)$ be a parameter, and $\Ex\bracket{\MM(\GEp)} \leq \mu$. For any $\lambda > 0$,  
	\begin{align*}
		\Pr\Paren{\card{\MM(\GEp) - \Ex\bracket{\MM(\GEp)}} \geq \lambda } \leq 2\cdot\exp\paren{-\frac{\lambda^2 \cdot p}{2 \cdot \mu}}.
	\end{align*}
\end{lemma}
\begin{proof}
	For simplicity, define $G_p := \GEp$. Let $C$ be any minimum vertex cover in the graph $G$. We use \emph{vertex exposure} martingales over vertices in $C$ to prove this result. 
	Fix an arbitrary ordering of vertices in $C$ and for any $v \in C$, define $C^{< v}$ as the set of vertices in $C$ that appear before $v$ in this ordering. For each $v \in C$, we define
	a random variable $Y_v \in \set{0,1}^{V \setminus C^{< v}}$ as a vector of indicators whether a possible edge (i.e., an edge already in $G$) between
	the vertices $v$ and $u \in V \setminus C^{<v}$ appears in $G_p$ or not. Since $C$ is a vertex cover of $G$, every edge in $G$ is incident on some vertex of $C$. As a result, 
	the graph $G_p$ is uniquely determined by the vectors $Y_1,\ldots,Y_{\card{C}}$. Define a sequence of random variables $\set{X_i}_{i=1}^{\card{C}}$, 
	whereby $X_i = \Ex\bracket{\MM(G_p) \mid Y_1,\ldots,Y_{i}}$. The following claim is standard. 
	
	\begin{claim}\label{clm:doob-martingale}
		The sequence $\set{X_i}_{i=1}^{\card{C}}$ is a martingale with respect to the sequence $\set{Y_i}_{i=1}^{\card{C}}$.  
	\end{claim}
	\begin{proof}
		For any $i \leq \card{C}$, 
		\begin{align*}
			\Ex\bracket{X_i \mid Y_1,\ldots,Y_{i-1}} &= \Ex_{(y_1,\ldots,y_{i-1})}\Bracket{\Ex\bracket{\MM(G_p) \mid Y_1=y_{1},\ldots,Y_{i-1}=y_{i-1},Y_i} \mid Y_1=y_1,\ldots,Y_{i-1}=y_{i-1}} \\
			&=  \Ex_{(y_1,\ldots,y_{i-1})}\Bracket{\MM(G_p)\mid Y_1=y_1,\ldots,Y_{i-1}=y_{i-1}} \tag{as we are ``averaging out'' $Y_i$ in the outer expectation} \\
			&= \Ex\Bracket{\MM(G_p) \mid Y_1,\ldots,Y_{i-1}} = X_{i-1}. \qed
		\end{align*}
		
	\end{proof}
	Notice that $X_0 := \Ex\bracket{\MM(G_p)}$ and $X_{\card{C}} = \MM(G_p)$ as fixing $Y_1,\ldots,Y_{\card{C}}$ uniquely determines the graph $G_p$. Hence, we can use Azuma's inequality  to show that
	value of $X_{\card{C}}$ is close to $X_0$ with high probability. To do this, we need a bound on $\card{C}$, as well as each
	 term $\card{X_i - X_{i-1}}$. Bounding each $\card{X_i - X_{i-1}}$ term is quite easy; the set of edges incident on the vertex $i$ can only change the 
	maximum matching in $G_p$ by $1$ (as $i$ can only be matched once), and hence $\card{X_i - X_{i-1}} \leq 1$. In the following, we also bound the value of $\card{C}$. 
	\begin{claim}\label{clm:vc-upper-bound}
		$\card{C} \leq 2\cdot \mu/p$. 
	\end{claim}
	\begin{proof}
		Since size of a minimum vertex cover of a graph $G$ is at most twice the size of its maximum matching (Fact~\ref{fact:vc-matching}), we have that $\card{C} \leq 2\cdot\MM(G)$. It is also straightforward to
		verify that $p \cdot \MM(G) \leq \Ex\Bracket{\MM(G_p)} \leq \mu$, since $p$ fraction of the edges
		of any maximum matching of $G$ appear in $G_p$ in expectation; hence $\card{C} \leq 2\cdot\mu/p$. 
	\end{proof}
	
	We are now ready to finalize the proof. By setting $c_i = 1$ for all $i \leq \card{C}$, we can use Azuma's inequality (Proposition~\ref{prop:azuma}) with
	parameters $\lambda$ and $c_i$ for the martingales $\set{X_i}$, and obtain that,
	\begin{align*}
		\Pr\Paren{\card{\MM(G_p) - \Ex\bracket{\MM(G_p)}} \geq \lambda } &= \Pr\Paren{\card{X_{\card{C}} - X_0} \geq \lambda} \Leq{Proposition~\ref{prop:azuma}} 2\cdot \exp\paren{-\frac{\lambda^2}{\sum_{i \in C} c_i^2}} \\
		&= 2\cdot \exp\paren{-\frac{\lambda^2}{\card{C}}} \Leq{Claim~\ref{clm:vc-upper-bound}} 2\cdot\exp\paren{-\frac{\lambda^2 \cdot p}{2 \cdot \mu}}
	\end{align*}
	finalizing the proof. 
\end{proof}

\subsection{Proof of Theorem~\ref{thm:maxmatching-3}}\label{sec:maxmatching-3}

Let $G(V,E)$ be any arbitrary graph and $\Gi{1},\ldots,\Gi{k}$ be a random $k$-partition of $G$. Recall that \MaxMatching coreset simply computes a maximum matching $M_i$ on each graph $\Gi{i}$ for $i \in [k]$; hence, we only need
to show that the graph $H(V,M_1 \cup \ldots \cup M_k)$ has a large matching compared to the graph $G$. 

Let $\Mstar$ be any fixed maximum matching in $G$, and let $\mu := \card{\Mstar} = \MM(G)$. Our approach is to show that either each graph $\Gi{i}$ has a large matching already, i.e., $\card{M_i} \geq \mu/3$, or 
many edges of $\Mstar$ are picked in $M_i$ as well. In the latter case, the union of edges in $M_i$ for $i \in [k]$ has a large intersection with $\Mstar$ and hence contains a large matching. 

Define $\Gm(V,\Em)$ whereby $\Em := E \setminus \Mstar$. Let $\Gm_i := \Gm \cap \Gi{i}$ be the intersection of the graph $G^{(i)}$ and $\Gm$. 
Finally, define $\mu^{-}_i$ as the \emph{maximum matching size} in $\Gm_i$. Using our concentration result from the previous section, we can show that, 

\begin{claim}\label{clm:maxmatching-size}
	Let $\eps \in (0,1)$ be a parameter. Suppose $\mu \geq 4\cdot\eps^{-2} \cdot k\log{n}$; then, there exists an integer $\mu^- \in [n]$ such that with probability $1-o(1)$ (over the random $k$-partition), $\mu^{-}_i = \mu^- \pm \eps \cdot \mu$ 
	simultaneously for all $i \in [k]$. 
\end{claim}
\begin{proof}
	Let $p = 1/k$; the graph $\Gm_i$ is a subgraph of $\Gm$ obtained by picking each edge in $\Gm$ independently and with probability $p$. Let $\mu^- := \Ex\bracket{\MM(\Gm_i)} \leq \mu$ (notice that
	the marginal distribution of $\Gm_i$ graphs for all $i \in [k]$ are identical). By setting $\lambda = \eps \cdot \mu$ in Lemma~\ref{lem:concentration}, we have that, 
	\begin{align*}
		\Pr\paren{\card{\mu^{-}_i - \mu^-} \geq \lambda} &\Leq{Lemma~\ref{lem:concentration}} 2 \cdot \exp\paren{-\frac{\eps^2 \cdot \mu^2 \cdot p}{2 \cdot \mu}} \leq 2 \cdot \exp \paren{-2\log{n}} \leq \frac{1}{n^2}
	\end{align*}
	where the second inequality is by the assumption on the value of $\mu$. Taking a union bound over all $k \leq n$ subgraphs $\Gm_i$ for $i \in [k]$ finalizes the proof. 
\end{proof}

In the following, we condition on the event in Claim~\ref{clm:maxmatching-size}. We now have,

\begin{lemma}\label{lem:maxmatching-intersect}
	Let $\mu$, $\mu^{-}$ and $\eps$ be as in Claim~\ref{clm:maxmatching-size}. If $\mu^{-} \leq \mu/3$, then $\card{\bigcup_{i=1}^{k} M_i \cap \Mstar} \geq \mu/3 - 3\eps \cdot \mu$ w.p. $1-o(1)$.  
\end{lemma}
\begin{proof}
	Fix an index $i \in [k]$ and notice that conditioning on the event in Claim~\ref{clm:maxmatching-size}, only fixes the set of edges in $\Gm_i$. Let $M^-_i$ be any maximum matching in $\Gm_i$; by definition, $\mu^-_i = \card{M^-_i}$. 
	By conditioning on the event in Claim~\ref{clm:maxmatching-size}, we have $\mu^-_i \leq \mu^- + 2\eps \cdot \mu$. It is straightforward 
	to verify that there are at least $\card{\Mstar} - 2\cdot\card{M^-_i} = \mu - \mu^-_i \geq \mu/3 - 2\eps\cdot\mu$ edges $e$ in $\Mstar$ such that neither of endpoints of $e$ are matched
	by $M^-_i$. We refer to these edges as \emph{free} edges and use $\Mstar_f \subseteq \Mstar$ to denote them. 
	
	 Note that even after conditioning on $\Gm_i$, the edges in $\Mstar$, and consequently $\Mstar_f$, appear in the graph $\Gi{i}$ independently and with probability $1/k$. As such, using a Chernoff bound (by assumption
	 on the value of $\mu$), w.p. $1-1/n^2$, $\card{\Mstar_f}/{k} - \eps \cdot \mu/k$ edges of $\Mstar_f$ appear in $\Gi{i}$. Since these edges can be directly added to the matching $M^-_i$ (as neither endpoints
	 of them are matched in $M^-_i$), this implies that there exists a matching of size $\mu^-_i + \frac{1}{k} \cdot \paren{\mu/3 - 3\eps \cdot \mu}$ in $\Gi{i}$ w.p. $1-1/n^2$. 
	 
	 Now let $M_i$ be the maximum matching computed by \MaxMatching; the above argument implies that $\card{M_i} \geq \mu^-_i + \frac{1}{k} \cdot \paren{\mu/3 - 3\eps \cdot \mu}$. On the other hand, notice that $\card{M_i \cap \Gm_i} \leq \mu^-_i$ 
	 as $M_i \cap \Gm_i$ forms a matching in the graph $\Gm_i$ and $\mu^-_i$ denotes the maximum matching size in this graph. This means that $\card{M_i \cap \Mstar} = \card{M_i \setminus \Gm_i} \geq \frac{1}{k} \cdot \paren{\mu/3 - 3\eps \cdot \mu}$. 
	 To finalize the proof, notice that by a union bound over all $k$ matchings $M_i$, we have that with probability $1-1/n$, 
	 \begin{align*}
	 	\card{\bigcup_{i=1}^{k} M_i \cap \Mstar} = \sum_{i=1}^{k} \card{M_i \cap \Mstar} \geq k \cdot \frac{1}{k} \cdot \paren{\mu/3 - 3\eps \cdot \mu} = \mu/3 - 3\eps \cdot \mu. \qed
	 \end{align*}
	 
\end{proof}

We can now easily prove Theorem~\ref{thm:maxmatching-3}. 

\begin{proof}[Proof of Theorem~\ref{thm:maxmatching-3}]
	By our assumption that $\mu = \MM(G) =  \omega(k\log{n})$, we can take $\eps$ in Claim~\ref{clm:maxmatching-size} and Lemma~\ref{lem:maxmatching-intersect} to be some arbitrary small constant, say $\eps = o(1)$. 
	Define $\mu^-$ as in Lemma~\ref{lem:maxmatching-intersect}. If $\mu^- > \mu/3$, we are already done as by Claim~\ref{clm:maxmatching-size}, for any $i \in [k]$, $\card{M_i} \geq \mu^-  - o(\mu) \geq \mu/3 - o(\mu)$ and hence 
	the union of matchings $M_1,\ldots,M_k$ surely has a $\paren{3+o(1)}$ approximate matching. On the other hand, if $\mu^- \leq \mu/3$, we can apply Lemma~\ref{lem:maxmatching-intersect}, and argue that $\mu/3 - o(\mu)$ edges
	of the matching $\Mstar$ appear in the union of matchings $M_1,\ldots,M_k$, which finalizes the proof. 
\end{proof}

\subsection{Lower Bound on Approximation Ratio of \MaxMatching}\label{sec:maxmatching-lower} 

We also show that there exists a graph for which the approximation ratio of \MaxMatching is arbitrarily close to $2$. This implies that we cannot improve the analysis of \MaxMatching much further and in particular 
beat the approximation ratio of $2$.  

\begin{lemma}\label{lem:maxmatching-lb}
	There exists a graph $G(V,E)$ such that for any random $k$-partition of $G$ ($k \leq n^{1-\delta}$ for any constant $\delta > 0$), 
	the \MaxMatching coreset can only find a matching of size at most $\paren{\frac{1}{2}+\frac{1}{k}} \cdot \MM(G)$ with high probability. 
\end{lemma}
\begin{proof}
	The vertex set of the graph $G$ consists of four sets of vertices $L_1,L_2,R_1,R_2$ with $\card{L_1} = \frac{n}{2} + \frac{n}{k}$ and $\card{L_2} = \card{R_1} = \card{R_2} = \frac{n}{2}$. $G$ is a bipartite graph 
	with $L_1 \cup L_2$ on one side of the bipartition and $R_1 \cup R_2$ on the other side. There is a complete bipartite graph between $L_1$ and $R_2$, a perfect matching between $L_2$ and $R_2$ and a matching of size 
	$\card{R_1}$ between $L_1$ and $R_1$. 
	
	It is easy to verify that there exists a matching of size $\card{R_1} + \card{R_2} = n$ in $G$ and hence $\MM(G) \geq n$. Suppose we create a random $k$-partition $\Gi{1},\ldots,\Gi{k}$ of $G$ and each machine
	$i \in [k]$ computes an arbitrary maximum matching $M_i$ of its input graph (i.e., compute the \MaxMatching coreset). In the following, we argue that the maximum matching in the 
	graph $H(V,M_1 \cup \ldots, M_k)$ is of size $(\frac{1}{2}+\frac{1}{k}) \cdot n$ with high probability, which concludes the proof. 
	
	To prove the lemma, we need the following simple claim about the maximum matching in the edge sampled subgraphs of $G$. 
	
	\begin{claim}\label{clm:maxmatching-lb}
		Suppose $\GEp(V,E_p)$ is an edge sampled subgraph of $G$ with probability $p = 1/k$; then, w.p. $1-1/n^5$, there exists a matching $M_p$ in $G$ such that: 
		\begin{enumerate}
			\item $M_p$ is a maximum matching in $G_p$, i.e., $\card{M_p} = \MM(G)$. 
			\item No edges between $L_2$ and $R_2$ belong to $M_p$. 
		\end{enumerate}
	\end{claim}
	\begin{proof}
		A simple application of Chernoff bound ensures that the total number of edges between $L_1$ and $R_1$ in $G_p$ is at most $2p\cdot n = n/k$ with probability at least $1-1/n^{10}$. In the following, we condition on this event. 
		Define $M_{1,1}$ as the matching consisting of the edges between $L_1$ and $R_1$ in $G_p$ and let $L^-_1 := L_1 \setminus L_1(M_{1,1})$ be the set of vertices in $L_1$ that are \emph{not} incident on $M_{1,1}$. 
		
		Consider the graph
		between $L^-_1$ and $R_2$. Note that since $\card{M_{1,1}} \leq n/k$, we have $\card{L^-_1} \geq \card{R_2}$. 
		By the independence in the sampling of edges and the fact that in $G$, $L^-_1$ and $R_1$ forms a bipartite clique, the set of edges between $L^-_1$ and $R_2$ in $G_p$ form a random bipartite graph
		with probability of having each edge equal to $p = 1/k = \omega(\log{n})$. Using standard facts about random graphs (see, e.g.,~\cite{Bollobas2001}, Chapter 7), 
		this implies that there exists a matching of size $\card{R_2}$ between $L^-_1$ and $R_2$ in $G_p$ with probability $1-1/n^{10}$. Let $M_p$ be the union of this matching and $M_{1,1}$. 
		
		It is clear that $M_p$ does not have any edges between $L_2$ and $R_2$. To see that $M_p$ is indeed a maximum matching of $G_p$, notice
		that all vertices in $R_1 \cup R_2$ in $G_p$ that have non-zero degree are matched by $M_p$ and hence there cannot be any larger matching in $G$. 
	\end{proof}
	
	We are now ready to finalize the proof of Lemma~\ref{lem:maxmatching-lb}. Recall that each graph $\Gi{i}$ is an edge sampled subgraph of $G$ with probability $1/k$. We can apply Claim~\ref{clm:maxmatching-lb} to each graph
	$\Gi{i}$ and by a union bound, w.p. $1-1/n^4$, there exists a suitable maximum matching $\Mi{i}_p$ in each graph $\Gi{i}$. Since we are choosing an \emph{arbitrary} maximum matching of $\Gi{i}$ as its coreset, we can assume that $\Mi{i}_p$ would 
	be chosen from each graph $\Gi{i}$, i.e., $M_i = \Mi{i}_p$ for all $i \in [k]$. This implies that no edge incident to vertices in $L_2$ are chosen among all coresets $M_1 \cup \ldots \cup M_k$. As a result, the maximum matching in the 
	graph $H(V,M_1 \cup \ldots \cup M_k)$ can have size at most $\card{L_1} = \paren{\frac{1}{2} + \frac{1}{k}} \cdot n$, finalizing the proof as $\MM(G) = n$. 
\end{proof}

%% file: edcs.tex
\section{New Properties of Edge Degree Constrained Subgraphs} \label{sec:edcs}

We study further properties of EDCS in this section. Although EDCS was used prior to our work, {all the properties proven in this section are entirely new to this paper} and look at the EDCS from a different vantage point. 

Previous work in~\cite{BernsteinS15,BernsteinS16} studied the EDCS from the perspective of how large of matching it contains and how it can be maintained efficiently in a dynamically changing graph. In this paper, 
we prove several new interesting structural properties of the EDCS itself. 
In particular, while it is easy to see that in terms of edge sets there
can be many different EDCS of some fixed graph $G(V,E)$ 
(consider $G$ being a complete graph), we show that the degree distributions
of every EDCS (for the same parameters $\beta$ and $\beta^-$) are almost the same.
In other words, the degree of any vertex $v$ is almost the same in every EDCS of $G(V,E)$.
This is in sharp contrast with similar objects such as $b$-matchings, which can vary a lot within the same graph. 
This semi-uniqueness renders the EDCS extremely robust under sampling and 
composition as we prove next in this section.  

These new structural results on EDCS are the main properties that allows their use in our coresets and parallel algorithms in the rest of the paper. In fact, our parallel algorithms 
in Section~\ref{sec:mpc-algs} are entirely based on these results and do not rely at all on the fact that an EDCS contains a large matching (i.e., do not depend on Lemma~\ref{lem:edcs-matching} at all).

\subsection{Degree Distribution Lemma} \label{sec:edcs-ddl}

%Consider any two EDCS $A$ and $B$ of a graph $G(V,E)$. 

%It is easy to see that in general an EDCS of a graph $G(V,E)$ may not be unique (consider $G$ being a complete graph for example). 
%Remarkably, we show that despite this fact, the degree distribution of every EDCS of any arbitrary graph $G(V,E)$ is (almost) unique. In other words,
%the degree of any vertex $v \in V$ is almost the same across any two different EDCS of a graph $G$. This uniqueness of degree distribution is the crucial property that allows us 
%to combine arbitrarily chosen EDCS of sampled subgraphs to form an EDCS of the original graph. 

\begin{lemma}[Degree Distribution Lemma]\label{lem:ddl}
	Fix a graph $G(V,E)$ and parameters $\beta,\beta^- = (1-\lambda) \cdot \beta$ (for $\lambda < 1/100$). For any two subgraphs $A$ and $B$ that are $\EDCS{G,\beta,\beta^-}$, and any vertex $v \in V$, 
	\[\card{\deg{A}{v} - \deg{B}{v}} = O(\log{n}) \cdot \lambda^{1/2} \cdot \beta. \]
\end{lemma}

	In the rest of this section, we fix the parameters $\beta,\beta^-$ and the two EDCS $A$ and $B$ in Lemma~\ref{lem:ddl}. The general strategy of the proof is as follows. We start with a set $S_1$ of all vertices which has 
	the most difference in degree between $A$ and $B$. By considering the two-hop neighborhood of these vertices in $A$ and $B$, we show that there exists a set $S_2$ of vertices in $V$ such that the difference between
	the degree of vertices in $A$ and $B$ is almost the same as vertices in $S_1$, while size of $S_2$ is a constant factor larger than $S_1$. We then use this argument repeatedly to construct the next set $S_3$ and so on, whereby each set
	is larger than the previous one by a constant factor, while the gap between the degree of vertices in $A$ and $B$ remains almost the same as the previous set. As this geometric increase in the size of sets can only happen in a ``small number'' of steps
	(otherwise we run out of vertices), we obtain that the gap between the degree of vertices in $S_1$ could have not been ``too large'' to begin with. 
  We now formalize this argument, starting with a technical lemma which allows us to 
obtain each set $S_i$ from the set $S_{i-1}$ in the above argument. 
	
\begin{lemma}\label{lem:ddl-increase}
	Fix an integer $D > 2\lambda^{1/2} \cdot \beta$ and suppose $S \subseteq V$ is such that for all $v \in S$, we have $\deg{A}{v} - \deg{B}{v} \geq D$. Then, there exists a set of 
	vertices $S' \supseteq S$ such that $\card{S'} \geq (1+2\lambda^{1/2}) \cdot \card{S}$ and for all $v \in S'$, $\deg{A}{v} - \deg{B}{v} \geq D - 2\lambda \cdot \beta$. 
\end{lemma}
\begin{proof}
	We define the following two sets $T$ and $T'$: 
	\begin{itemize}
		\item $T$ is the set of all neighbors of vertices in $S$ using \emph{only} the edges in $A \setminus B$. In other words, $T := \set{v \in V \mid \exists u \in S \wedge (u,v) \in A \setminus B}$.
		\item $T'$ is the set of all neighbors of vertices in $T$ using \emph{only} the edges in $B \setminus A$. In other words, $T' := \set{w \in V \mid \exists v \in T \wedge (v,w) \in B \setminus A}$. 
	\end{itemize}

	We start by proving the following property on the degree of vertices in the sets $T$ and $T'$. %Roughly speaking, we expect the vertices in $T$ to have a higher degree in $B$ compared to $A$, and similarly vertices in $T'$
	\begin{claim}\label{clm:ddl-T}
		We have,
		\begin{itemize}
			\item for all $v \in T$,  $\deg{B}{v} - \deg{A}{v} \geq D - \lambda \cdot \beta$. 
			\item for all $w \in T'$, $\deg{A}{w} - \deg{B}{w} \geq D - 2\lambda \cdot \beta$. 
		\end{itemize}
	\end{claim}
	\begin{proof}
		For the first part, since $v \in T$, it means that there exists an edge $(u,v) \in A \setminus B$ such that $u \in S$. Since $(u,v)$ belongs to $A$, by $\propone$ of an EDCS we have 
		$\deg{A}{v} \leq \beta - \deg{A}{u} \leq \beta - \deg{B}{u} - D$. On the other hand, since $(u,v)$ does not belong to $B$, by $\proptwo$ of an EDCS we have
		$\deg{B}{v} \geq \beta - \lambda \cdot \beta - \deg{B}{u}$, completing the proof for vertices in $T$. 
		
		For the second part, since $w \in T'$, it means that there exists an edge $(v,w) \in B \setminus A$ such that $v \in T$. Since $(v,w)$ does not belong to $A$, by $\proptwo$ of an EDCS we have 
		$\deg{A}{w} \geq \beta - \lambda \cdot \beta - \deg{A}{v}$.  Moreover, since $(u,v)$ belongs to $B$, by $\propone$ of an EDCS, we have, $\deg{B}{w} \leq \beta  - \deg{B}{v}$. 
		This means that $\deg{A}{w} - \deg{B}{w} \geq \deg{B}{v} - \deg{A}{v} - \lambda \cdot \beta$ which is at least $D - 2\lambda \cdot \beta$ by the first part. 
	\end{proof}
	
	Notice that since $D > 2\lambda \cdot \beta$, by Claim~\ref{clm:ddl-T}, for any vertex $v \in T$, we have $\deg{B}{v} > \deg{A}{v}$ and hence $S \cap T = \emptyset$ (similarly, $T \cap T' = \emptyset$, but $S$ and $T'$ may intersect). 
	We define the set $S'$ in the lemma statement to be $S' := S \cup T'$. The bound on the degree of vertices in $S'$ follows immediately from Claim~\ref{clm:ddl-T} (recall that vertices in $S$ already satisfy the degree requirement for the set $S'$). 
	In the following, we show that $\card{T' \setminus S} \geq 2\lambda^{1/2} \cdot \card{S}$, which finalizes the proof. 
	
	Recall that  $E_{A \setminus B}(S)$ and $E_{A \setminus B}(S,T)$ denote the set of edges in subgraph $A \setminus B$ incident on vertices $S$, and between vertices $S$ and $T$, respectively. 
	We have, 
	\begin{align*}
		\card{E_{B \setminus A}(T,T' \setminus S)} &= \card{E_{B \setminus A}(T)} - \card{E_{B \setminus A}(T,S)} \tag{as all the edges in $B \setminus A$ that are incident on $T$ are going to $T'$}\\
		&\geq \card{E_{A \setminus B}(T)} - \card{E_{B \setminus A}(T,S)}  \tag{as by Claim~\ref{clm:ddl-T}, the degree of vertices in $T$ is larger in $B \setminus A$ compared to $A \setminus B$} \\
		&\geq \card{E_{A \setminus B}(S)} - \card{E_{B \setminus A}(S)} \tag{as all edges in $A \setminus B$ incident on $S$ are also incident on $T$}  \\
		&\geq \card{S} \cdot D \tag{by the assumption on the degree of vertices in $S$ in subgraphs $A$ and $B$}
	\end{align*}
	
	Finally, since $B$ is an EDCS, the maximum degree of any vertex in $T' \setminus S$ is at most $\beta$ and hence there should be at least $\card{S} \cdot \frac{D}{\beta} \geq 2\lambda^{1/2} \cdot \card{S}$ vertices in $T' \setminus S$  
	(as $D > 2\lambda^{1/2}\cdot \beta$). 
\end{proof}

%We are now ready to prove Lemma~\ref{lem:ddl}. 

\begin{proof}[Proof of Lemma~\ref{lem:ddl}]
	Suppose towards a contradiction that there exists a vertex $v \in V$ s.t. $D := \deg{A}{v} - \deg{B}{v} \geq 3\ln{(n)} \cdot \lambda^{1/2}\cdot \beta$ (the other case is symmetric). Let $D_0 = D$ and $S_0 = \set{v}$ and 
	for $i = 1$ to $t:=\lambda^{-1/2} \cdot \paren{{\ln{(n)}+1}}$: define the set $S_i$ and integer $D_i$ by applying Lemma~\ref{lem:ddl-increase} to $S_{i-1}$ and $D_{i-1}$ (i.e., $S_i = S'$ and $D_i = D_{i-1} - 2\lambda \cdot \beta$). 
	By the lower bound on the value of $D$, for any 
	$i \in [t]$, we have that $D_i \geq D - i \cdot 2\lambda \cdot \beta > 2\lambda^{1/2} \cdot \beta$, and hence we can indeed apply Lemma~\ref{lem:ddl-increase}. As a result, we have, 
	\begin{align*}
		\card{S_{t}} \geq \paren{1+2\lambda^{1/2}} \cdot \card{S_{t-1}} \geq \paren{1+2\lambda^{1/2}}^{t} \cdot \card{S_0} \geq \exp{\paren{\lambda^{1/2} \cdot t}} > \exp\paren{\ln{(n)}} = n.
	\end{align*}
	which is a contradiction as there are only $n$ vertices in the graph $G$. Consequently, we obtain that for any vertex $v$, $\card{\deg{A}{v} - \deg{B}{v}} = O(\log{n}) \cdot \lambda^{1/2} \cdot \beta$, finalizing the proof.  
	%, whereby for all $v \in S_i$, $\deg{A}{v} - \deg{B}{v} \geq D_i$ with $D_i := D_{i-1} - 2\lambda \cdot \beta$. 
\end{proof}

\subsection{EDCS in Sampled Subgraphs} \label{sec:edcs-sampled}

In this section, we prove two lemmas regarding the structure of different EDCS across sampled subgraphs. The first lemma concerns edge sampled subgraphs. We show that the degree distributions of any two EDCS for 
two different edge sampled subgraphs of $G$ is almost the same no matter how the two EDCS are selected or even if the choice of the two subgraphs are \emph{not} independent. 
This Lemma is is used in our Result \ref{res:coresets} on randomized coresets (see Section \ref{sec:coresets-matching-vc}).

\begin{lemma}[EDCS in Edge Sampled Subgraphs] \label{lem:edcs-edge-sampled}
	Fix any graph $G(V,E)$ and $p \in (0,1)$. Let $G_1$ and $G_2$ be two \emph{edge sampled subgraphs} of $G$ with probability $p$ (chosen \emph{not} necessarily independently). Let $H_1$ and $H_2$ be 
	arbitrary EDCSs of $G_1$ and $G_2$ with parameters $\paren{\beta, (1-\lambda) \cdot \beta}$. Suppose $\beta \geq 750 \cdot \lambda^{-2} \cdot \ln{(n)}$, then, with probability $1-4/n^9$, simultaneously for all $v \in V$: 
	\begin{align*}
		\card{\deg{H_1}{v} - \deg{H_2}{v}} \leq O(\log{n}) \cdot \lambda^{1/2} \cdot \beta. 
	\end{align*}
\end{lemma}

We also prove a qualitatively similar lemma for vertex sampled subgraphs. This is needed in Result \ref{res:mpc} for parallel algorithms (see Section \ref{sec:mpc-algs}). The main difference here is that there will be a huge gap between the degree of a vertex between the two EDCS if the vertex is sampled in one subgraph but not the other one. However, 
we show that the degree of vertices that are sampled in both subgraphs are almost the same across the two different (and arbitrarily chosen) EDCS for the subgraphs.  

\begin{lemma}[EDCS in Vertex Sampled Subgraphs]\label{lem:edcs-vertex-sampled}
	Fix any graph $G(V,E)$ and $p \in (0,1)$. Let $G_1$ and $G_2$ be two \emph{vertex sampled subgraphs} of $G$ with probability $p$ (chosen \emph{not} necessarily independently). Let $H_1$ and $H_2$ be 
	arbitrary EDCSs of $G_1$ and $G_2$ with parameters $\paren{\beta, (1-\lambda)\beta}$. If $\beta \geq 750 \cdot \lambda^{-2} \cdot \ln{(n)}$, then, with probability $1-4/n^9$, simultaneously for all $v \in G_1 \cap G_2$: 
	\begin{align*}
		\card{\deg{H_1}{v} - \deg{H_2}{v}} \leq O(\log{n}) \cdot \lambda^{1/2} \cdot \beta. 
	\end{align*}
\end{lemma}

The proof of both these lemmas are along the following lines. We start with an EDCS $H$ of the original graph $G$ with parameters (almost) $\paren{\beta/p,\beta^-/p}$.  We then consider the set of edges from $H$ in each of the sampled
subgraphs $G_1$ and $G_2$, i.e., the two subgraphs $H'_1 := G_1 \cap H$ and $H'_2 := G_2 \cap H$. We use the randomness in the process of sampling subgraphs $G_1$ and $G_2$ to prove that with high probability both $H'_1$ and $H'_2$ form an
EDCS for $G_1$ and $G_2$, respectively, with parameters $\paren{\beta,\beta^-}$. Finally, we use our degree distribution lemma from Section~\ref{sec:edcs-ddl} to argue that for any arbitrary EDCS $H_1$ (resp. $H_2$) of $G_1$ (resp. $G_2$), the degree 
distribution of $H_1$ (resp. $H_2$) is close to $H'_1$ (resp. $H'_2$). Since the degree distributions of $H'_1$ and $H'_2$ are close to each other already (as they are both sampled subgraphs of $H$), this finalizes the proof. 

There are some technical differences in implementing the above intuition between the edge sampled and vertex sampled subgraphs and hence we provide a separate proof for each case.

\subsubsection{EDCS in Edge Sampled Subgraphs: Proof of Lemma~\ref{lem:edcs-edge-sampled}}
\begin{proof}[Proof of Lemma~\ref{lem:edcs-edge-sampled}]
	We first prove that edge sampling an EDCS results in another EDCS for the sampled subgraph. 
	 
	\begin{claim}\label{clm:edcs-edge-sampled}
		Let $H$ be an $\EDCS{G,\beta_H,\beta^-_H}$ for parameters $\beta_H := (1-\frac{\lambda}{2}) \cdot \frac{\beta}{p}$ and $\beta^-_H := \beta_H-1$. 
		Suppose $G_p := G^E_p(V,E_p)$ is an edge sampled subgraph of $G$ and $H_p := H \cap G_p$; then, with probability $1-2/n^9$: 
		\begin{enumerate}
			\item For any vertex $v \in V$, $\card{\deg{H_p}{v} - p \cdot \deg{H}{v}} \leq \frac{\lambda}{5} \cdot \beta$.
			\item $H_p$ is an EDCS of $G_p$ with parameters $\paren{\beta,(1-\lambda) \cdot \beta}$.  
		\end{enumerate}
	\end{claim}
	\begin{proof}
	For any vertex $v \in V$, $\Ex{\bracket{\deg{H_p}{v}}} = p \cdot \deg{H}{v}$ and 
	$\deg{H}{v} \leq \beta_H$ by $\propone$ of EDCS $H$. Moreover, since each neighbor of $v$ in $H$ is sampled in $H_p$ independently, by Chernoff bound (Proposition~\ref{prop:chernoff}), we have, 
	\begin{align*}
		\Pr\paren{\card{\deg{H_p}{v} - p \cdot \deg{H}{v}} \geq \frac{\lambda}{5} \cdot \beta} \leq 2\cdot \exp\paren{-\frac{\lambda^2  \cdot \beta}{75}}  \leq 2\cdot \exp\paren{-10\ln{n}} = \frac{2}{n^{10}},
	\end{align*}
	where the second inequality is by the lower bound on $\beta$ in Lemma~\ref{lem:edcs-vertex-sampled} statement. 
	In the following, we condition on the event that: 
	\begin{align}
		\forall v \in V~~ \card{\deg{H_p}{v} - p \cdot \deg{H}{v}} \leq \frac{\lambda}{5} \cdot \beta. \label{eq:edge-degree-event}
	\end{align}
	This event happens with probability at least $1-2/n^9$ by above equation and a union bound on $\card{V} = n$ vertices. This finalizes the proof of the first part of the claim. 
	We are now ready to prove that $H_p$ is indeed an $\EDCS{G_p,\beta,(1-\lambda)\cdot \beta}$ conditioned on this event. 
	
	Consider any edge $(u,v) \in H_p$. Since $H_p \subseteq H$, $(u,v) \in H$ as well. Hence, we have, 
	\begin{align*}
		\deg{H_p}{u} + \deg{H_p}{v} \Leq{Eq~(\ref{eq:edge-degree-event})} p \cdot \Paren{\deg{H}{u} + \deg{H}{v}} + \frac{2\lambda}{5} \cdot \beta \leq p \cdot \beta_H +  \frac{2\lambda}{5} \cdot \beta  
		= (1-\frac{\lambda}{2}) \cdot \beta + \frac{2\lambda}{5} \cdot \beta < \beta,
	\end{align*}
	where the second inequality is by \propone of EDCS $H$ and the equality is by the choice of $\beta_H$. As a result, $H_p$ satisfies \propone of EDCS for parameter $\beta$.
	
	Now consider an edge $(u,v) \in G_p \setminus H_p$. Since $H_p = G_p \cap H$,  $(u,v) \notin H$ as well. Hence,  
	\begin{align*}
		\deg{H_p}{u} + \deg{H_p}{v} &\Geq{Eq~(\ref{eq:edge-degree-event})} p \cdot \Paren{\deg{H}{u} + \deg{H}{v}} - \frac{2\lambda}{5} \cdot  \beta \geq p \cdot \beta^-_H - \frac{2\lambda}{5} \cdot  \beta \\
		&= (1-\frac{\lambda}{2}) \cdot \beta - p - \frac{2\lambda}{5} \cdot \beta 
		> (1-\lambda) \cdot \beta,
	\end{align*}
	where the second inequality is by \proptwo of EDCS $H$ and the equality is by the choice of $\beta^-_H$. 
	As such, $H_p$ satisfies \proptwo of EDCS for parameter $(1-\lambda) \cdot \beta$ and hence $H_p$ is indeed an $\EDCS{G_p,\beta,(1-\lambda) \cdot \beta}$. 
	\end{proof}
	
	We continue with the proof of Lemma~\ref{lem:edcs-edge-sampled}. 
	%We can now use Claim~\ref{clm:edcs-vertex-sampled} together with Lemma~\ref{lem:ddl} to finalize the proof. 
	Let $H$ be an $\EDCS{G,\beta_H,\beta^-_H}$ for the parameters $\beta_H,\beta_H^-$ in Claim~\ref{clm:edcs-edge-sampled}. The existence of $H$ follows from Lemma~\ref{lem:edcs-exists} as $\beta_H^- < \beta_H$. 
	Define $\hH_{1} := H \cap G_1$ and $\hH_{2} := H \cap G_2$. By Claim~\ref{clm:edcs-edge-sampled}, $\hH_{1}$ (resp. $\hH_{2}$) is an EDCS of $G_1$ (resp. $G_2$) with parameters $\paren{\beta,(1-\lambda)\beta}$ with
	probability $1-4/n^9$. In the following, we condition on this event. 
	 
	By Lemma~\ref{lem:ddl} (Degree Distribution Lemma), since both $H_{1}$ (resp. $H_2$) and $\hH_1$ (resp. $\hH_2$) are EDCS for $G_1$ (resp. $G_2$), the degree of vertices in both of them should be ``close'' to each other. Moreover, since
	by Claim~\ref{clm:edcs-vertex-sampled} the degree of each vertex in $\hH_1$ and $\hH_2$ is close to $p$ times its degree in $H$, we can argue that the vertex degrees in $H_1$ and $H_2$ are close. Formally, 
	for any $v \in V$, we have, 
	\begin{align*}
		\card{\deg{H_1}{v} - \deg{H_2}{v}} &\leq \card{\deg{H_1}{v} - \deg{\hH_1}{v}} + \card{\deg{\hH_1}{v} - \deg{\hH_2}{v}} + \card{\deg{\hH_2}{v} - \deg{H_2}{v}}  \\
		&\!\!\!\!\!\!\!\!\!\Leq{Lemma~\ref{lem:ddl}} O(\log{n}) \cdot \lambda^{1/2} \cdot \beta + \card{\deg{\hH_1}{v} - p \cdot \deg{H}{v}} + \card{\deg{\hH_2}{v} - p \cdot \deg{H}{v}} \\
		&\!\!\!\!\!\!\!\!\Leq{Claim~\ref{clm:edcs-edge-sampled}} O(\log{n}) \cdot \lambda^{1/2} \cdot \beta + O(1) \cdot \lambda \cdot \beta,
	\end{align*}
	finalizing the proof. 
\end{proof}

\subsubsection{EDCS in Vertex Sampled Subgraphs: Proof of Lemma~\ref{lem:edcs-vertex-sampled}}
\begin{proof}[Proof of Lemma~\ref{lem:edcs-vertex-sampled}]

	We first prove that vertex sampling an EDCS results in another EDCS for the sampled subgraph. 
	 
	\begin{claim}\label{clm:edcs-vertex-sampled}
		Let $H$ be an $\EDCS{G,\beta_H,\beta^-_H}$ for parameters $\beta_H := (1-\frac{\lambda}{2}) \cdot \frac{\beta}{p}$ and $\beta^-_H := \beta_H - 1$. 
		Suppose $G_p := G^V_p(V_p,E_p)$ is a vertex sampled subgraph of $G$ and $H_p := H \cap G_p$; then, with probability $1-2/n^9$: 
		\begin{enumerate}
			\item For any vertex $v \in V_p$, $\card{\deg{H_p}{v} - p \cdot \deg{H}{v}} \leq \frac{\lambda}{5} \cdot \beta$.
			\item $H_p$ is an EDCS of $G_p$ with parameters $\paren{\beta,(1-\lambda) \cdot \beta}$.  
		\end{enumerate}
	\end{claim}
	\begin{proof}
	For any vertex $v \in V_p$, $\Ex{\bracket{\deg{H_p}{v}}} = p \cdot \deg{H}{v}$ by the independent sampling of vertices and 
	$\deg{H}{v} \leq \beta_H$ by \propone of EDCS $H$. Moreover, since each neighbor of $v$ in $H$ is sampled in $H_p$ independently, by Chernoff bound (Proposition~\ref{prop:chernoff}), we have, 
	\begin{align*}
		\Pr\paren{\card{\deg{H_p}{v} - p \cdot \deg{H}{v}} \geq \frac{\lambda}{5} \cdot \beta} \leq 2\cdot \exp\paren{-\frac{\lambda^2  \cdot \beta}{75}}  \leq 2\cdot \exp\paren{-10\ln{n}} = \frac{2}{n^{10}},
	\end{align*}
	where the second inequality is by the lower bound on $\beta$ in Lemma~\ref{lem:edcs-vertex-sampled} statement. 
	In the following, we condition on the event that: 
	\begin{align}
		\forall v \in V_p~~ \card{\deg{H_p}{v} - p \cdot \deg{H}{v}} \leq \frac{\lambda}{5} \cdot \beta. \label{eq:vertex-degree-event}
	\end{align}
	which happens with probability at least $1-2/n^9$ by above equation and a union bound on $\card{V_p} \leq n$ vertices. This finalizes the proof of the first part of the claim. 
	We are now ready to prove that $H_p$ is indeed an $\EDCS{G_p,\beta,(1-\lambda)\cdot \beta}$ conditioned on this event. 
	
	Consider any edge $(u,v) \in H_p$. Since $H_p \subseteq H$, $(u,v) \in H$ as well. Hence, we have, 
	\begin{align*}
		\deg{H_p}{u} + \deg{H_p}{v} \Leq{Eq~(\ref{eq:vertex-degree-event})} p \cdot \Paren{\deg{H}{u} + \deg{H}{v}} + \frac{2\lambda}{5} \cdot \beta \leq p \cdot \beta_H +  \frac{2\lambda}{5} \cdot \beta  
		= (1-\frac{\lambda}{2}) \cdot \beta + \frac{2\lambda}{5} \cdot \beta < \beta,
	\end{align*}
	where the second inequality is by \propone of EDCS $H$ and the equality is by the choice of $\beta_H$. As a result, $H_p$ satisfies \propone of EDCS for parameter $\beta$.
	
	Now consider an edge $(u,v) \in G_p \setminus H_p$. Since $H_p = G_p \cap H$,  $(u,v) \notin H$ as well. Hence,  
	\begin{align*}
		\deg{H_p}{u} + \deg{H_p}{v} &\Geq{Eq~(\ref{eq:vertex-degree-event})} p \cdot \Paren{\deg{H}{u} + \deg{H}{v}} - \frac{2\lambda}{5} \cdot  \beta \geq p \cdot \beta^-_H - \frac{2\lambda}{5} \cdot  \beta \\
		&= (1-\frac{\lambda}{2}) \cdot \beta - p - \frac{2\lambda}{5} \cdot \beta 
		> (1-\lambda) \cdot \beta,
	\end{align*}
	where the second inequality is by \proptwo of EDCS $H$ and the equality is by the choice of $\beta^-_H$. 
	As such, $H_p$ satisfies \proptwo of EDCS for parameter $(1-\lambda) \cdot \beta$ and hence $H_p$ is indeed an $\EDCS{G_p,\beta,(1-\lambda) \cdot \beta}$. 
	\end{proof}
	
	We continue with the proof of Lemma~\ref{lem:edcs-vertex-sampled}. 
	%We can now use Claim~\ref{clm:edcs-vertex-sampled} together with Lemma~\ref{lem:ddl} to finalize the proof. 
	Let $H$ be an $\EDCS{G,\beta_H,\beta^-_H}$ for the parameters $\beta_H,\beta_H^-$ in Claim~\ref{clm:edcs-vertex-sampled}. The existence of $H$ follows from Lemma~\ref{lem:edcs-exists} as $\beta^-_H < \beta_H$. 
	Define $\hH_{1} := H \cap G_1$ and $\hH_{2} := H \cap G_2$. By Claim~\ref{clm:edcs-vertex-sampled}, $\hH_{1}$ (resp. $\hH_{2}$) is an EDCS of $G_1$ (resp. $G_2$) with parameters $\paren{\beta,(1-\lambda)\beta}$ with
	probability $1-4/n^9$. In the following, we condition on this event. 
	 
	By Lemma~\ref{lem:ddl} (Degree Distribution Lemma), since both $H_{1}$ (resp. $H_2$) and $\hH_1$ (resp. $\hH_2$) are EDCS for $G_1$ (resp. $G_2$), the degree of vertices in both of them should be ``close'' to each other. Moreover, since
	by Claim~\ref{clm:edcs-vertex-sampled} the degree of each vertex in $\hH_1$ and $\hH_2$ is close to $p$ times its degree in $H$, we can argue that the degree of shared vertices in $H_1$ and $H_2$ are close. Formally, 
	let $v$ be a vertex in both $G_1$ and $G_2$; we have, 
	\begin{align*}
		\card{\deg{H_1}{v} - \deg{H_2}{v}} &\leq \card{\deg{H_1}{v} - \deg{\hH_1}{v}} + \card{\deg{\hH_1}{v} - \deg{\hH_2}{v}} + \card{\deg{\hH_2}{v} - \deg{H_2}{v}}  \\
		&\!\!\!\!\!\!\!\!\!\Leq{Lemma~\ref{lem:ddl}} O(\log{n}) \cdot \lambda^{1/2} \cdot \beta + \card{\deg{\hH_1}{v} - p \cdot \deg{H}{v}} + \card{\deg{\hH_2}{v} - p \cdot \deg{H}{v}} \\
		&\!\!\!\!\!\!\!\!\Leq{Claim~\ref{clm:edcs-vertex-sampled}} O(\log{n}) \cdot \lambda^{1/2} \cdot \beta + O(1) \cdot \lambda \cdot \beta,
	\end{align*}
	finalizing the proof. 
\end{proof}

%% file: coresets.tex
\section{Randomized Coresets for Matching and Vertex Cover}\label{sec:coresets-matching-vc}

We introduce our randomized coresets for matching and vertex cover in this section. Both of these results are achieved by computing an EDCS of the input graph (for appropriate choice of parameters) 
and then applying Lemmas~\ref{lem:edcs-matching} and~\ref{lem:edcs-vc}. % to argue that this EDCS allows for computing a large matching and a small vertex cover of the input graph. 

\subsection{Computing an EDCS from Random $k$-Partitions}\label{sec:edcs-coreset} 

Let $G(V,E)$ be any arbitrary graph and $\Gi{1},\ldots,\Gi{k}$ be a random $k$-partition of $G$. We show that if we compute an arbitrary EDCS of each graph $\Gi{i}$ (with no coordination across different graphs) and 
combine them together, we obtain an EDCS for the original graph $G$. %Consider the following process: 
\begin{tbox}
	\begin{enumerate}
		\item Let $\Gi{1},\ldots,\Gi{k}$ be a random $k$-partition of the graph $G$. 
		\item For any $i \in [k]$, compute $\Ci{i} := \EDCS{G,\beta, (1-\lambda) \cdot \beta}$ for parameters 
		\[ 
			\lambda = \Theta\paren{(\frac{\eps}{\log{n}})^2} ~~~~~~\textnormal{and}~~~~~~ \beta := \Theta(\lambda^{-3} \cdot \log{n}). 
		\] 
		\item Let $C := \bigcup_{i=1}^{k} \Ci{i}$. 
	\end{enumerate}
\end{tbox}

%In the following, we prove that the subgraph $C$ computed in this process is an EDCS of $G$ for parameters (roughly) $(k \cdot \beta, k \cdot (1-\lambda) \cdot \beta)$. 

\begin{lemma}\label{lem:edcs-coreset}
	With probability $1-4/n^{7}$, the subgraph $C$ is an $\EDCS{G,\beta_C,\beta^{-}_C}$ for parameters: 
	\[
		\lambda_C:= O(\log{n}) \cdot \lambda^{1/2}~~~~~~\textnormal{,}~~~~~~ \beta_C := (1+\lambda_C) \cdot k \cdot \beta ~~~~~~\textnormal{and}~~~~~~ \beta^-_C := (1-2\lambda_C) \cdot k \cdot \beta. 
	\]	
\end{lemma}
\begin{proof}
	Recall that each graph $\Gi{i}$ is an edge sampled subgraph of $G$ with sampling probability $p = \frac{1}{k}$. 
	By Lemma~\ref{lem:edcs-edge-sampled} for graphs $\Gi{i}$ and $\Gi{j}$ (for $i \neq j \in [k]$) and their EDCSs $\Ci{i}$ and $\Ci{j}$, with probability $1-4/n^9$, 
	for all vertices $v \in V$: 
	\begin{align}
	\card{\deg{\Ci{i}}{v} - \deg{\Ci{j}}{v}} \leq O(\log{n}) \cdot \lambda^{1/2} \cdot \beta = \lambda_C \cdot \beta. \label{eq:edge-pairwise-degrees}
	\end{align}
	By taking a union bound on all ${{k} \choose{2}} \leq n^2$ pairs of subgraphs $\Gi{i}$ and $\Gi{j}$ for $i \neq j \in [k]$, the above property holds for all $i,j \in [k]$, with probability at least $1-4/n^{7}$. In the following, 
	we condition on this event. 
	
	We now prove that $C$ is indeed an $\EDCS{G,\beta_C,\beta^-_C}$. 
	First, consider an edge $(u,v) \in C$ and let $j \in [k]$ be such that $(u,v) \in \Ci{j}$ as well. We have,
	\begin{align*}
		\deg{C}{u} + \deg{C}{v} &= \sum_{i=1}^{k} \deg{\Ci{i}}{u} + \sum_{i=1}^{k} \deg{\Ci{i}}{v} \Leq{Eq~(\ref{eq:edge-pairwise-degrees})} k \cdot \paren{\deg{\Ci{j}}{u} + \deg{\Ci{j}}{v}} + k \cdot \lambda_C \cdot \beta \\
		&\leq k \cdot \beta + k \cdot \lambda_C \beta = \beta_C. \tag{by \propone of EDCS $\Ci{j}$ with parameter $\beta$}
	\end{align*}
	Hence, $C$ satisfies \propone of EDCS for parameter $\beta_C$. 
	
	Now consider an edge $(u,v) \in G \setminus C$ and let $j \in [k]$ be such that $(u,v) \in \Gi{j} \setminus \Ci{j}$ (recall that each edge in $G$ is sent to exactly one graph $\Gi{j}$ in the random $k$-partition). We have,
	\begin{align*}
		\deg{C}{u} + \deg{C}{v} &= \sum_{i=1}^{k} \deg{\Ci{i}}{u} + \sum_{i=1}^{k} \deg{\Ci{i}}{v} \Geq{Eq~(\ref{eq:edge-pairwise-degrees})} k \cdot \paren{\deg{\Ci{j}}{u} + \deg{\Ci{j}}{v}} - k \cdot \lambda_C \cdot \beta \\
		&\geq k \cdot (1-\lambda) \cdot \beta - k \cdot \lambda_C \cdot \beta \geq (1-2\lambda_C) \cdot k \cdot \beta = \beta^-_C. \tag{by \proptwo of EDCS $\Ci{j}$ with parameter $(1-\lambda) \cdot \beta$}
	%	&\geq t \cdot (1-\lambda) \cdot \beta - t \cdot O(\log{n}) \cdot \lambda^{1/2} \cdot \beta \geq t \cdot (1-\lambda_C) \cdot \beta = \beta^-_C. \tag{by the second property of EDCS $\Ci{j}$ with parameter $(1-\lambda) \beta$}
	\end{align*}
	Hence, $C$ also satisfies the second property of EDCS for parameter $\beta^-_C$, finalizing the proof. 
\end{proof}

\subsection{EDCS as a Coreset for Matching and Vertex Cover}\label{sec:matching-vc-coreset}

We are now ready to present our randomized coresets for matching and vertex cover using the EDCS as the coreset, formalizing Result~\ref{res:coresets}. %Formally,  

\begin{theorem}\label{thm:matching-edcs-coreset}
	Let $G(V,E)$ be a graph and $\Gi{1},\ldots,\Gi{k}$ be a random $k$-partition of $G$. For any $\eps \in (0,1)$, any $\EDCS{\Gi{i},\beta,(1-\lambda) \cdot \beta}$ for $\lambda := \Theta\paren{(\frac{\eps}{\log{n}})^2}$ 
	and $\beta := \Theta(\eps^{-6} \cdot \log^{7}{n})$ is a $\paren{3/2 + \eps}$-approximation randomized composable coreset of size $O(n \cdot \beta)$ for the maximum matching problem. 
\end{theorem}
\begin{proof}
	By Lemma~\ref{lem:edcs-coreset}, the union of the coresets, i.e., the $k$ EDCSs, is itself an $\EDCS{G,\beta_C,\beta^{-}_C}$, such that $\beta^{-}_C = (1-\Theta(\eps)) \cdot \beta_C$. 
	Hence, by Lemma~\ref{lem:edcs-matching}, the maximum matching in this EDCS is of size $\paren{2/3-\eps} \cdot \MM(G)$. The bound on the size of the coreset follows from $\propone$ of EDCS 
	as maximum degree in the EDCS computed by each machine is at most $\beta$ and hence size of each coreset is $O(n \cdot \beta) = \Ot_{\eps}(n)$. 
\end{proof}

To present our coreset for the vertex cover problem, we need to slightly relax the definition of randomized coreset. Following~\cite{AssadiK17}, we augment the definition of randomized coresets by allowing the coresets to also contain a fixed solution (which is 
counted in the size of the coreset) to be directly added to the final solution of the composed coresets. In other words, the coreset contains both subsets of vertices (to be always included in the final vertex cover) and edges (to guide the choice of additional vertices 
in the vertex cover). This definition is necessary for the vertex cover problem due to the hard-to-verify feasibility constraint of this problem; see~\cite{AssadiK17} for more details.

\begin{theorem}\label{thm:vc-edcs-coreset}
	Let $G(V,E)$ be a graph and $\Gi{1},\ldots,\Gi{k}$ be a random $k$-partition of $G$. For any $\eps \in (0,1)$, any $\EDCS{\Gi{i},\beta,(1-\lambda) \cdot \beta}$ for $\lambda := \Theta\paren{(\frac{\eps}{\log{n}})^2}$ 
	and $\beta := \Theta(\eps^{-6} \cdot \log^{7}{n})$ plus the set of vertices with degree larger than $(1-\Theta(\eps)) \cdot \beta/2$ in the EDCS (to be added directly to the final vertex cover) 
	is a $\paren{2 + \eps}$-approximation randomized composable coreset of size $O(n \cdot \beta)$ for the minimum vertex cover problem. 
\end{theorem}
\begin{proof}
	By Lemma~\ref{lem:edcs-coreset}, the union of the coresets, i.e., the $k$ EDCSs, is itself an $\EDCS{G,\beta_C,\beta^{-}_C}$ $C$, such that $\beta^{-}_C = (1-\Theta(\eps)) \cdot \beta_C$. 
	Suppose first that instead of each coreset fixing the set of vertices to be added to the final vertex cover, we simply add all vertices with degree more than $\beta^{-}_C/2$ to the vertex cover and then compute
	a minimum vertex cover of $C$. In this case, by Lemma~\ref{lem:edcs-vc}, the returned solution is a $(2+\eps)$-approximation to the minimum vertex cover of $G$. 
	
	To complete the argument, recall that the degree of any vertex $v \in V$ is 
	essentially the same across all machines (up to an additive term of $\eps \cdot \beta$) by Lemma~\ref{lem:edcs-edge-sampled}, and hence the set of vertices with degree more than $\beta^{-}_C/2$ would be a subset of 
	the set of fixed vertices across all machines. Moreover, any vertex added by any machine to the final vertex cover has degree at least $(1-\Theta(\eps)) \cdot \beta^{-}_C/2$ and hence we can apply Lemma~\ref{lem:edcs-vc}, with 
	a slightly smaller parameter $\eps$ to argue that the returned solution is still a $(2+\eps)$-approximation. 
\end{proof}

\begin{remark}\label{rem:vc-polytime}
	In the proof of Theorem~\ref{thm:vc-edcs-coreset}, we neglected the time necessary to compute a vertex cover in the union of the coresets (as is consistent with the definition of randomized coresets). 
	In case we require this algorithm to run in polynomial time, we need to approximate the final vertex cover in the union of the coresets as opposed to recover it exactly. In particular, by picking 
	a $2$-approximation vertex cover in the union of coreset in the proof of Lemma~\ref{lem:edcs-vc}, we obtain an (almost) $4$-approximation to the vertex cover of $G$. 
\end{remark}

%% file: mpc-algorithms.tex
\section{MPC Algorithms for Matching and Vertex Cover}\label{sec:mpc-algs}
%%\todo[inline]{For the shortened version, I would simply remove this entire paragraph, as well as the repetition of the Theorem.}
%%As was shown in Section~\ref{sec:coreset}, our randomized coresets for matching and vertex cover in Section~\ref{sec:coresets-matching-vc} immediately imply MPC algorithms for these problems in only two MPC rounds. 
%%However, while quite round-efficient, the memory requirement of these algorithms can be as large as $\Theta(n\sqrt{n})$ in case of dense input graphs. 

In this section, we show that a careful adaptation of our coresets construction together 
with the structural results proven for EDCS in Section~\ref{sec:edcs} can be used to obtain MPC algorithms with much smaller memory while increasing the number of required rounds to only $O(\log\log{n})$. 

\begin{theorem}\label{thm:mpc-matching-vc}
	There exists an MPC algorithm that given a graph $G(V,E)$ with high probability computes an $O(1)$ approximation to both maximum matching and minimum vertex cover of $G$ in $O(\log\log{n} + \log{\paren{\frac{n}{s}})}$ MPC rounds
	 on machines of memory $s = n^{\Omega(1)}$.
\end{theorem}

By setting $s = O(n/\polylog{(n)})$ in Theorem~\ref{thm:mpc-matching-vc}, we achieve an $O(1)$-approximation algorithm to both matching and vertex cover in $O(\log\log{n})$ MPC rounds on machines of memory $O(n/\polylog{(n)})$, formalizing Result~\ref{res:mpc}. 

In the following, for the sake of clarity, we mostly focus on proving Theorem~\ref{thm:mpc-matching-vc} for the natural case when memory per machine is $s = \Ot(n)$, and postpone the proof for all range
of parameter $s$ to Section~\ref{sec:mpc-smaller-memory}. 
The overall idea of our algorithm is as follows. 
Instead of the edge sampled subgraphs used by our randomized coresets, we start by picking $k = O(n)$ \emph{vertex} sampled subgraphs of $G$ with sampling probability roughly $1/\sqrt{n}$ 
and send each to a separate machine. 
Each machine then locally computes an EDCS of its input (with parameters $\beta = \polylog{(n)}$ and $\beta^- \approx \beta$) with no coordination across the machines. Unlike the MPC algorithm obtained 
by our randomized coreset approach (Corollary~\ref{cor:coreset-mpc}), where the memory per machine was as large as $\Theta(n\sqrt{n})$, here we cannot collect all these smaller EDCSes on a single machine of memory $\Ot(n)$. 
Instead, we repartition them across the machines again (and discard remaining edges) and repeat the previous process on this new graph. The main 
observation is that after each step, the maximum degree of the remaining graph (i.e., the union of all EDCSes) would drop quadratically (e.g., from potentially $\Omega(n)$ to $\Ot(\sqrt{n})$ in the first step). As such, in each subsequent step, 
we can pick a smaller number of vertex sampled subgraphs, each with a higher sampling probability than previous step, and still each graph fits into the memory of a single machine. Repeating this process for $O(\log\log{n})$ steps reduces the maximum degree of the remaining graph
to $\polylog{(n)}$. At this point, we can store the final EDCS on a single machine and solve the problem locally. 

Unfortunately this approach on its own would only yield a $(3/2)^{O(\log\log{n})} = \polylog{(n)}$ approximation to matching, since by Lemma \ref{lem:edcs-matching} each recursion onto an EDCS
of the graph could introduce a $3/2$-approximation.
A similar problem exists for vertex cover. In the proof of 
Lemma \ref{lem:edcs-vc}, computing a vertex cover of $G$ from its EDCS $H$ involves two steps: we
add to the vertex cover all vertices with high degree in $H$ to cover the edges in $G \setminus H$,
and then we separately compute a vertex cover for the edges in $H$. Since $H$ cannot fit into a single machine, the second computation is done recursively: in each round, we find an EDCS of the current graph (which is partitioned amongst many machines), add 
to the vertex cover all high degree vertices in this EDCS, and then recurse onto the sparser EDCS. A straightforward analysis would only lead
to an $O(\log\log{n})$ approximation. 

We improve the approximation factor for both vertex cover
and matching by showing that they can serve as witnesses to each other. Every time we
add high-degree vertices to the vertex cover, we will also find a large matching incident to these
vertices: we show that this can be done in $O(1)$ parallel rounds. We then argue
that their sizes are always within a constant factor of each other, so both are a constant
approximation for the respective problem (by Proposition~\ref{prop:vc-matching-alpha}).

The rest of this section is organized as follows. We first present our subroutine for computing the EDCS of an input graph in parallel using vertex sampled subgraphs. 
Next, we present a simple randomized algorithm for finding a large matching incident on high degree vertices of an input graph. Finally, we combine these two subroutines to provide our main parallel algorithm for approximating matching and vertex cover. 
We finish this section 
by specifying the MPC implementation of our parallel algorithm and finalize the proof of Theorem~\ref{thm:mpc-matching-vc}.

\subsection{A Parallel Algorithm for EDCS} \label{sec:mpc-edcs}

We now present our parallel algorithm for computing an EDCS via vertex sampling.
In this algorithm, the edges of the input graph as well as the output EDCS will be partitioned across multiple machines.
%Each machine has $\Ot(n)$ memory, and we assume that there are at least $\polylog{(n)}$ machines.
In the following, we use a slightly involved method of sampling the vertices using limited independence. This is due to technical reasons 
in the MPC implementation of this algorithm which we explain in Remark~\ref{rem:limited-independence}.  
To avoid repeating the arguments, we present our algorithm for all range of memory $s=n^{\Omega(1)}$, but encourage the reader to consider the case of $s = n$ for more intuition. 

%We start by presenting our parallel algorithm for computing an EDCS of a graph $G$ in parallel over machines of memory $\Ot(n)$ each. Instead of the edge sampled subgraphs used 
%in Section~\ref{sec:edcs-coreset}, here we combine EDCSs computed on vertex sampled subgraphs. Moreover, as opposed to the case for randomized coresets where the final EDCS resides on a 
%single machine, in this case the computed EDCS is still edge partitioned across multiple machines. 
%In the following, we assume that the number of machines in $\pedcs$ is at least some $\polylog{(n)}$. 

\textbox{\textnormal{$\pedcs(G,\Delta,s)$. A parallel algorithm for EDCS of a graph $G$ with maximum degree $\Delta$ on machines of memory $\Ot(s)$.}}{
	\begin{enumerate}
		\item Define $p = (200\log{n}) \cdot \sqrt{\frac{s}{n \cdot \Delta}}$ and $k = \frac{800\log{n}}{p^2}$.
		\item Create $k$ vertex sampled subgraphs $\Gi{1},\ldots,\Gi{k}$ on $k$ different machines as follows: 
			\begin{enumerate}
				\item Let $\kappa := (20\log{n})$. Each vertex $v$ in $G$ \emph{independently} picks a \emph{$\kappa$-wise independent hash
				 function} $h_v: [k] \rightarrow [1/p]$. 
				 \item The graph $\Gi{i}$ is the induced subgraph of $G$ on vertices $v \in V$ with $h_v(i) = 0$. 
			\end{enumerate}
		\item Define parameters $\lambda := \paren{2\cdot \log{n}}^{-3}$ and $\beta := 750 \cdot \lambda^{-2} \cdot \ln{(n)}$. 
		
		\item For $i = 1$ to $k$ {in parallel}:  {Compute} $\Ci{i} = \EDCS{\Gi{i},\beta,(1-\lambda)\cdot \beta}$ locally on machine $i$.

		\item Define the \emph{multi-graph} $C(V,E_C)$ with $E_C := \bigcup_{i=1}^{k} \Ci{i}$ (allowing for multiplicities). Notice that this multi-graph is edge partitioned across the machines. 
	\end{enumerate}
}

For any vertex $v \in V$, define $I(v) \subseteq k$ as the set of indices of the subgraphs that sampled vertex $v$. Notice that indices 
in $I(v)$ are $\kappa$-wise independent random variables. Additionally, it is easy to see that each graph $\Gi{i}$ is a vertex sampled subgraph of $G$ with sampling probability $p$.

\begin{remark}\label{rem:limited-independence} 
\emph{
As opposed to the previous vertex sampling approach of Czumaj~\etal~\cite{CzumajLMMOS17} that resulted in a partitioning of vertices of $G$ across different subgraphs, our way of sampling subgraphs in $\pedcs$
	results in each vertex appearing in $\Theta(p\cdot k)$ different subgraphs with high probability. This is necessary for our algorithm as we need to ensure that every edge of the input graph is sampled in this process. 
	However, this property introduces new challenges in the 
	MPC implementation of our algorithm as a naive implementation of this idea requires communicating $\Theta(p \cdot k)$ messages per each edge of the graph which cannot be done within the memory restrictions of the MPC model. 
	This is the main reason that we sample these subgraphs in $\pedcs$ with limited independence as opposed to truly independently to reduce the communication necessary per each edge to $O(\log{n})$. }
\end{remark}

We first prove some simple properties of \pedcs. 
\begin{proposition}\label{prop:pedcs-simple}
	For $\Delta \geq \paren{\frac{n}{s}} \cdot \paren{400 \cdot \log^{12}{(n)}}$, with probability $1-2/n^{8}$, 
	\begin{enumerate}
		\item For any vertex $v \in V$, $\card{I(v)} = p \cdot k \pm \lambda \cdot p \cdot k$. 
		\item For any edge $e \in E$, there exists at least one index $i \in [k]$ such that $e$ belongs to $\Gi{i}$. 
	\end{enumerate}
\end{proposition}
\begin{proof}
	Fix any vertex $v \in V$. Clearly, $\Ex\card{I(v)} = p \cdot k$. Moreover, $\card{I(v)}$ is sum of zero-one $\kappa$-wise independent random variables
	and hence by Chernoff bound with bounded independence (Proposition~\ref{prop:bounded-chernoff})
	\begin{align*}
		\Pr\paren{\card{\card{I(v)} - \Ex\card{I(v)}} \geq \lambda \cdot \Ex\card{I(v)}} \leq 2 \cdot \exp\paren{-\frac{\kappa}{2}} \leq 2 \cdot \exp\paren{-10\log{n}} \leq 1/n^{10}. 
	\end{align*}
	Note that $\Ex\card{I(v)} = p \cdot k = \Theta\paren{\sqrt{\frac{n \cdot \Delta}{s}}}$, $\lambda = (2 \cdot \log{n})^{-3}$ and hence by the bound on $\Delta$, we have $\lambda^{2} \cdot \Ex\card{I(v)}/2 = \Theta(\log{n}) = \kappa$, and hence
	we can indeed apply Proposition~\ref{prop:bounded-chernoff} here. By a union bound on all $n$ vertices, the first part holds w.p. $\geq 1-1/n^{9}$. 
	
	We now prove the second part. Fix an edge $e \in E$ and define the indicator random variables $X_1,\ldots,X_k$ where $X_i = 1$ iff $e$ is contained in the graph $\Gi{i}$. 
	Define $X:= \sum_{i=1}^{k} X_i$ to denote the number of graphs the edge $e$ belongs to. Clearly, $\Ex\bracket{X_i} = p^{2}$ for all $i \in [k]$ and hence $\Ex\bracket{X} = p^2 \cdot k$. 
	Moreover, the random variables $X_i$'s are $\kappa$-wise independent for $\kappa = 20\log{n}$. Hence, by Chernoff bound with bounded independence (Proposition~\ref{prop:bounded-chernoff}), 
	the probability that $e$ belongs to no graph $\Gi{i}$, i.e., $X = 0$ is at most, 
	\begin{align*}
		\Pr\paren{X = 0} \leq \Pr\Paren{\card{X - \Ex\bracket{X}} \geq \Ex\bracket{X}} \leq  2\cdot \exp\paren{-\frac{\kappa}{2}} \leq 2\cdot \exp\paren{-10\log{n}} \leq 1/n^{10}.
	\end{align*}
	Again, note that $\Ex\card{X} = p^{2} \cdot k = 800\log{n} = 2\kappa$ and hence we could apply Proposition~\ref{prop:bounded-chernoff}. 
	By a union bound on all $O(n^2)$ edges, the second part also holds w.p. at least $1-1/n^{8}$. Another union bound on this event and the event in the first part finalizes the proof. 
\end{proof}

We now prove that the graph $C$ defined in the last line of $\pedcs$ is also an EDCS of $G$ with appropriate parameters. The proof is quite similar to that of Lemma~\ref{lem:edcs-coreset} with some additional 
care to handle the difference between vertex sampled subgraphs and edge sampled ones. 

\begin{lemma}\label{lem:pedcs-correctness}
	For $\Delta \geq \paren{\frac{n}{s}} \cdot \paren{400 \cdot \log^{12}{(n)}}$, with probability $1-5/n^7$, $C$ is an $\EDCS{G,\beta_C,\beta^-_C}$ for parameters: 
	\[
	\lambda_C :=  \lambda^{1/2} \cdot \Theta(\log{n}) = o(1),~~~~ \beta_C:=  p \cdot k \cdot (1+\lambda_C) \cdot \beta, ~~\textnormal{and}~~ \beta^-_C = p \cdot k \cdot (1-\lambda_C) \cdot \beta.
	\]  
\end{lemma}
\begin{proof}
	Recall that each graph $\Gi{i}$ is a vertex sampled subgraph of $G$ with sampling probability $p$. Hence, 
	by Lemma~\ref{lem:edcs-vertex-sampled} and a union bound, with probability $1-4/n^{7}$, for any two subgraphs $\Ci{i}$ and $\Ci{j}$ for $i,j \in [k]$, and any vertex $v \in \Vi{i} \cap \Vi{j}$, we have, 
	\begin{align}
		\card{\deg{\Ci{i}}{v} - \deg{\Ci{j}}{v}} \leq O(\log{n}) \cdot \lambda^{1/2} \cdot \beta. \label{eq:pairwise-degrees}
	\end{align}

	In the following, we condition on the events in Eq~(\ref{eq:pairwise-degrees}) and Proposition~\ref{prop:pedcs-simple} which happen together with probability at least $1-5/n^7$. 
	
	We now prove that $C$ is indeed an $\EDCS{G,\beta_C,\beta^-_C}$ for the given parameters. 
	%For any vertex $v \in V$, let $I(v) \subseteq [k]$ denotes the set of $t = \card{I(v)}$ indices such that $v$ is assigned to $\Vi{i}$ for $i \in I(v)$. 
	First, consider an edge $(u,v) \in C$ and let $j \in [k]$ be such that $(u,v) \in \Ci{j}$ as well. 
	We have,
	\begin{align*}
		\deg{C}{u} + \deg{C}{v} &= \sum_{i \in I(u)} \deg{\Ci{i}}{u} + \sum_{i \in I(v)} \deg{\Ci{i}}{v} \\
		&\Leq{Eq~(\ref{eq:pairwise-degrees})} \card{I(u)} \cdot \deg{\Ci{j}}{u} + \card{I(v)} \cdot \deg{\Ci{j}}{v} + \paren{\card{I(u)} + \card{I(v)}} 
		\cdot O(\log{n}) \cdot \lambda^{1/2} \cdot \beta \\
		&\!\!\!\!\!\!\!\!\!\Leq{Proposition~\ref{prop:pedcs-simple}} p \cdot k \cdot \beta + O(\lambda) \cdot p \cdot k \cdot \beta + p \cdot k \cdot O(\log{n}) \cdot \lambda^{1/2} \cdot \beta \tag{by \propone of EDCS $\Ci{j}$ with parameter $\beta$} \\
		&~~\leq p \cdot k \cdot (1+\lambda_C) \cdot \beta = \beta_C. 
	\end{align*}
	Hence, $C$ satisfies \propone of EDCS for parameter $\beta_C$. 
	
	Now consider an edge $(u,v) \in G \setminus C$ and let $j \in [k]$ be such that $(u,v) \in \Gi{j} \setminus \Ci{j}$ (the existence of $j$ follows from conditioning on the event in Proposition~\ref{prop:pedcs-simple}). We have, 
	\begin{align*}
		\deg{C}{u} + \deg{C}{v} &= \sum_{i \in I(u)} \deg{\Ci{i}}{u} + \sum_{i \in I(v)} \deg{\Ci{i}}{v} \\
		&\Geq{Eq~(\ref{eq:pairwise-degrees})} \card{I(u)} \cdot \deg{\Ci{j}}{u} + \card{I(v)} \cdot \deg{\Ci{j}}{v} - \paren{\card{I(u)} + \card{I(v)}} 
		\cdot O(\log{n}) \cdot \lambda^{1/2} \cdot \beta \\
		&\!\!\!\!\!\!\!\!\!\Geq{Proposition~\ref{prop:pedcs-simple}} p \cdot k \cdot (1-\lambda) \cdot \beta - O(\lambda) \cdot p \cdot k \cdot \beta - p \cdot k \cdot O(\log{n}) \cdot \lambda^{1/2} \cdot \beta \tag{by \proptwo of EDCS $\Ci{j}$ with parameter $\beta$} \\
		&~~\geq p \cdot k \cdot (1-\lambda_C) \cdot \beta = \beta^-_C. 
	\end{align*}
	Hence, $C$ also satisfies \proptwo of EDCS for parameter $\beta^-_C$, finalizing the proof. 
\end{proof}

Before moving on, we prove a simple claim that ensures that the memory of $\Ot(s)$ per machine in $\pedcs$ is enough for storing each subgraph $\Gi{i}$ and computing $\Ci{i}$ locally. 
	
\begin{claim}\label{clm:mpc-pedcs-size}
		With probability $1-1/n^{18}$, the total number of edges in each subgraph $\Gi{i}$ of $G$ in $\pedcs(G,\Delta,s)$ is $O(s \cdot \log^{2}{n})$. 
\end{claim}
	\begin{proof}
		Let $v$ be a vertex in $\Gi{i}$. By the independent sampling of vertices in a vertex sampled subgraph, we have that $\Ex\bracket{\deg{\Gi{i}}{v}} = p \cdot \deg{G}{v} \leq p \cdot \Delta = \Theta(\sqrt{\frac{s \cdot \Delta}{n}} \cdot \log{n})$. 
		By Chernoff bound, with probability $1-1/n^{20}$, degree of $v$ is $O(\sqrt{\frac{s \cdot \Delta}{n}} \cdot \log{n})$. We can then take a union bound on all vertices in $\Gi{i}$ and have that with probability $1-1/n^{19}$, the maximum
		 degree of $\Gi{i}$ is $O(\sqrt{\frac{s \cdot \Delta}{n}} \cdot \log{n})$. 
		At the same time, the expected number of vertices sampled in $\Gi{i}$ is at most $p \cdot n = \Theta(\sqrt{\frac{s \cdot n}{\Delta}} \cdot \log{n})$. 
		Another application of Chernoff bound ensures that the total number of vertices sampled in $\Gi{i}$ is $O(\sqrt{\frac{s \cdot n}{\Delta}} \cdot \log{n})$
		with probability $1-1/n^{19}$. As a result, the total number of edges in $\Gi{i}$ is $O(\sqrt{\frac{s \cdot \Delta}{n}} \cdot \log{n}) \cdot O(\sqrt{\frac{s \cdot n}{\Delta}} \cdot \log{n}) = O(s \cdot \log^{2}{n})$ 
		with probability at least $1-1/n^{18}$. 
	\end{proof}
\subsection{Random Match Algorithm}\label{sec:random-match}

In our main algorithm, we need a subroutine for finding a large matching incident on the set of ``high'' degree vertices of a given graph $G$ which its edges are initially partitioned across many machines. 
In this section, we provide such an algorithm based on a simple randomized procedure that is easily implementable in constant number of MPC rounds. 

\textbox{\textnormal{$\rmatch(G,S,\Delta)$. A parallel algorithm for finding a matching $M$ incident on given vertices $S$ in a graph $G$ with maximum degree $\Delta$. }}{

	\begin{enumerate}
		\item Sample each vertex in $S$ with probability $1/2$ independently to obtain a set $S'$. 
		\item For each vertex in $S'$ pick one of its incident edges to $G \setminus {S'}$ uniformly at random. Let $\Esample$ be the set of these edges.
		\item Let $M$ be the matching in $\Esample$ consists of all edges with unique endpoints; these are edges $(u,v) \in \Esample$ such that neither $u$ nor $v$ are incident on any other edge of $\Esample$. 
	\end{enumerate}
}

We prove that if the set $S$ consists of high degree vertices of $G$, then $\rmatch(G,S,\Delta)$ finds a large matching in $S$. Formally, 

\begin{lemma}\label{lem:random-match}
	Suppose $G(V,E)$ is a graph with maximum degree $\Delta \geq 100\log{n}$ and $S \subseteq V$ is such that for all $v \in S$, $\deg{G}{v} \geq \Delta/3$. 
	The size of the matching $M := \rmatch(G,S,\Delta)$ is in expectation $\Ex\card{M} = \Theta(\card{S})$.
\end{lemma}
\begin{proof}
	Fix any vertex $v \in S'$; we argue that with high probability, degree of $v$ to vertices in $G \setminus S'$ is at least $\Delta/7$.  This follows immediately as in expectation, at most half of the neighbors of $v$ belong to $S'$ and 
	we can apply Chernoff bound as $\Delta \geq 100\log{n}$. We apply a union bound on all vertices in $S'$ and in the following we condition on the event that all these vertices have at least $\Delta/7$ edges to $G \setminus S'$, which
	 happens with high probability. 
	
	By construction, any vertex in $S'$ has degree exactly one in $\Esample$. As such, to lower bound the size of $M$, we only need to lower bound the number of vertices in $G \setminus S'$ that have degree exactly one in $\Esample$. 
	Fix a vertex $v \in S'$ and consider the neighbor $u \in G \setminus S'$ of $v$ in $\Esample$. We know that $u$ has at most $\Delta-1$ other neighbors in $S'$ and each of these neighbors are choosing $u$ with probability at most $7/\Delta$ (as each of 
	them has at least $\Delta/7$ neighbors). Hence, 
	\begin{align*}
		\Pr\paren{\textnormal{$u$ has degree $1$ in $\Esample$}} \geq \paren{1-\frac{7}{\Delta}}^{\Delta-1} = \Theta(1). 
	\end{align*}
	
	As such, in expectation, $\Theta(S)$ vertices in $G \setminus S'$ also have degree exactly one in $\Esample$, which implies $\Ex\card{M} = \Theta(\card{S})$. 
\end{proof}

\subsection{A Parallel Algorithm for Matching and Vertex Cover}\label{sec:mpc-main}

We now present our main parallel algorithm. For sake of clarity, we present and analyze our algorithm here for the case when the memory allowed per each machine is $\Ot(n)$. In Section~\ref{sec:mpc-smaller-memory}, we show
how to easily extend this algorithm to the case when memory per machine is $O(s)$ for any choice of $s = n^{\Omega(1)}$. 

\textbox{\textnormal{$\palg(G,\Delta)$. A parallel algorithm for computing a vertex cover $\va$ and a matching $\ma$ of a given graph $G$ with maximum degree at most $\Delta$. }}{

	\begin{enumerate}
		\item If $\Delta \leq \paren{400 \cdot \log^{12}{n}}$ send $G$ to a single machine and run the following algorithm locally: 
		Compute a maximal matching $\ma$ in $G$ and let $\va$ be the set of vertices matched by $\ma$. 
		Return $\va$ and $\ma$. %Otherwise, continue with the following algorithm. 
		\item If $\Delta > \paren{400 \cdot \log^{12}{n}}$, we run the following algorithm. 
		\item Compute an EDCS $C := \pedcs(G,\Delta,n)$ in parallel. Let $\beta_C,\beta^-_C$ be the parameters of this EDCS (as specified in Claim~\ref{clm:palg-pedcs} below). 
		  %compute an $\EDCS{G,\beta_C,\beta^-_C}$ (for parameters $\beta_C,\beta^-_C$ defined in Lemma~\ref{lem:edcs-vertex-partition}) denoted by $$.
		\item Define $\vh := \set{v \in V \mid \deg{C}{v} \geq \beta^-_C/2}$ be the set of ``high'' degree vertices in $C$. 
		\item Compute a matching $\mh := \rmatch(C,\vh,\beta_C)$. 
		\item Define $\vm := V \setminus \Paren{\vh \cup V(\mh)}$ as the set of vertices that are neither high degree in $C$ nor matched by $\mh$. Let $\cm$ be the induced subgraph of $C$ on vertices $\vm$ with parallel edges removed.
		\item Recursively compute $(\vr,\mr) := \palg(\cm, \beta_C)$. 
		\item Return $\va := \vh \cup V(\mh) \cup \vr$ and $\ma := \mh \cup \mr$. 
	\end{enumerate}	
}

We start by proving some simple properties of $\palg$. The following claim is a direct corollary of Lemma~\ref{lem:pedcs-correctness} by setting $s = n$. 
\begin{claim}\label{clm:palg-pedcs}
	The subgraph $C: = \pedcs(G,\Delta,n)$ computed in $\palg(G,\Delta)$ is an $\EDCS{G,\beta_C,\beta^-_C}$ for parameters: 
	\[ 
		\lambda_C := o(1) ~~~~~ \beta_C := \sqrt{\Delta} \cdot O(\log^{5}{n}) ~~~~~ \beta^-_C := (1-\lambda_C) \cdot \beta_C  
	\]
	with probability at least $1-1/n^5$. 
\end{claim}

Similarly, the following claim follows easily from Lemma~\ref{lem:random-match}.
\begin{claim}\label{clm:palg-rmatch}
	Conditioned on $C=\EDCS{G,\beta_C,\beta^-_C}$, matching $\mh = \rmatch(C,\vh,\beta_C)$ has expected size $\Ex\card{\mh} = \Omega(\card{\vh})$. 
\end{claim}
\begin{proof}
	With this conditioning, the maximum degree of $G$ is at most $\beta_C$, while the degree of vertices $\vh$ is at least $\beta^-_C/2 \geq \beta_C/3$. Hence, we can apply
	Lemma~\ref{lem:random-match} and prove the statement. 
\end{proof}

Let $T$ be the number of recursive calls made by $\palg(G,\Delta)$. We refer to any $t \in [T]$ as a \emph{step} of $\palg$. We bound the total number of steps in $\palg$ as follows. 
\begin{claim}\label{clm:palg-steps}
	The total number of steps made by $\palg(G,\Delta)$ is $T = O(\log\log{\Delta})$. 
\end{claim}
\begin{proof}
	Define a function $F(\Delta)$ denoting the number of recursive calls made by $\palg$ with maximum degree $\Delta$. As $\palg(G,\Delta)$ runs $\palg(\cm,\beta_C)$ for $\beta_C < \Delta^{2/3}$, we have, 
	$F(\Delta) \leq F(\Delta^{2/3}) + 1$ for $\Delta > \paren{400\cdot \log^{12}{n}}$ and $F(\Delta) = 1$ otherwise. 
	It is now easy to see that $F(\Delta) = O(\log\log{\Delta})$, finalizing the proof. 
\end{proof}

In each step, $\palg$ runs
the subroutines $\pedcs$ and $\rmatch$ once. We say that a run of $\pedcs$ is \emph{valid} in this step iff the high probability event in Claim~\ref{clm:palg-pedcs} happens. 
Roughly speaking, this means that $\pedcs$ is valid when it returns the ``correct'' output. Additionally, we say that a step of $\palg$ is valid if $\pedcs$ subroutine in this step is valid.  
We define $\event$ as the event that all $T$ steps of $\palg(G,\Delta)$ are valid. By Claims~\ref{clm:palg-pedcs} each step of $\palg$ is valid with probability at 
least $1-1/n^{5}$. As there are in total $T= O(\log\log{n})$ steps by Claim~\ref{clm:palg-steps}, $\event$ happens with probability at least $1-1/n^{4}$. 

We are now ready to prove the correctness of $\palg$. 
\begin{lemma}\label{lem:palg-correct}
	For any graph $G(V,E)$, $\palg(G,n)$ with constant probability outputs a matching $\ma$ which is an $O(1)$-approximation to the maximum matching of $G$ and a vertex cover $\va$ which is an $O(1)$-approximation
	to the minimum vertex cover of $G$. 
\end{lemma}
\begin{proof}
	It is clear that the second parameter in $\palg(G,n)$ is an upper bound on the maximum degree of $G$ and hence $G$ satisfies the requirement of $\palg$. In the following, we condition on the event $\event$ which happens with high probability
	by the above discussion. As such, we also have that any recursive call to $\palg(C^-,\beta_C)$ is valid (i.e., $\beta_C$ is indeed an upper bound on degree of $C^-$) simply
	 because $C^-$ is a subgraph of an EDCS and hence its maximum degree is bounded by $\beta_C$. 
	
	We first argue that $\va$ and $\ma$ are respectively a feasible vertex cover and a feasible matching of $G$. The case for $\ma$ is straightforward; the set of vertices matched by $\mh$ is disjoint from the vertices in 
	$\mr$ as all vertices matched by $\mh$ are removed in $\cm$, and hence (by induction) $\ma = \mh \cup \mr$ is a valid matching in $G$. Now consider the set of vertices $\va$. By conditioning on the event $\event$, 
	$C$ is indeed an $\EDCS{G,\beta_C,\beta^-_C}$. Hence, by \proptwo of EDCS $C$, any edge $e \in G \setminus C$ has at least one neighbor in $\vh$ and is thus covered by $\vh$. 
	Additionally, as we pick $V(\mh)$ in the vertex cover, any edge incident on these vertices are also covered. This implies that $\vh \cup V(\mh)$ plus any vertex cover of the remaining graph $C^-$ is a feasible
	vertex cover of $G$. As $\vr$ is a feasible vertex cover of $C^-$ by induction, we obtain that $\va$ is also a feasible vertex cover of $G$ 
	(the analysis for the base case in step 1 where a maximal matching is compute locally is trivial).

	We now show that sizes of $\ma$ and $\va$ are within a constant factor of each other with constant probability. By Proposition~\ref{prop:vc-matching-alpha} this implies that both are an $O(1)$-approximation to their respective problem. 
	At each step, the set of vertices added to the $\va$ are of size $\card{V(\mh)} + \card{\vh} \leq 3\card{\vh}$ (as $\mh$ is incident on $\vh$). The set of edges added to matching $\ma$ are of size $\mh$ which is in expectation
	equal to $\Theta(\card{\vh})$ by Claim~\ref{clm:palg-rmatch}. As such, by induction and linearity of expectation, 
	this implies that $\Ex{\card{\ma}} = \Theta(\card{\va})$ (the base case is again trivial). To conclude, we can apply a Markov bound (on size of $\card{\va} - \card{\ma}$)
	and obtain that with constant probability $\card{\ma} = \Theta(\card{\va})$, which finalizes the proof. 
\end{proof}

We note that in Lemma~\ref{lem:palg-correct}, we only achieved a constant factor probability of success for $\palg$. We can however run this algorithm in parallel $O(\log{n})$ times and pick the best solution to achieve a high probability of success while still 
having $\Ot(n)$ memory per machine and $O(\log\log{n})$ rounds.

\subsection{MPC Implementation of the Parallel Algorithm}\label{sec:mpc-implementation} 

In this section, we briefly specify the details in implementing \palg in the MPC model on machines of memory $\Ot(n)$. In Section~\ref{sec:mpc-smaller-memory}, we show
how to extend this to the case when memory per machine is a given parameter $s$. Throughout this section, we assume that the event $\event$ defined in the previous section holds 
and hence we are implicitly conditioning on this (high probability) event. 

Our implementation is based on using by now standard tools in the MPC model for sorting and search in parallel introduced originally by~\cite{GoodrichSZ11} as specified in~\cite{CzumajLMMOS17}. 
On machines with memory $n^{\Omega(1)}$, the sort operation in~\cite{GoodrichSZ11} allows us to sort a set of key-value pairs of size $\polylog{(n)}$ in $O(1)$ MPC rounds. We can also do a parallel search:
given a set $A$ of key-value pairs and a set of queries each containing a key of an element in $A$, we can annotate each query with the corresponding key-value pair from $A$, again in $O(1)$ MPC rounds. 

We follow the approach of~\cite{CzumajLMMOS17} by using these operations to broadcast information from vertices to their incident edges. We build a collection of key-value pairs, where each key is a vertex and the value is the corresponding information. 
Then, each edge $(u,v)$ may issue two queries to obtain the information associated with $u$ and $v$. For more details, we refer the reader to Section~6 in~\cite{CzumajLMMOS17}. 
The following lemma states the main properties of our implementation. 

%%The proof of this lemma is based on quite similar ideas as 
%%MPC implementation of the algorithm of~\cite{CzumajLMMOS17} and hence in the proof, we simply state 

%%
%%
%%The first step is to specify the details of storing the graph $G(V,E)$ that allows for the functionalities required by our algorithm. 
%%
%%To each vertex $v \in V$, we dedicate a single machine, namely machine $\MV{v}$, of memory $O(n)$ that stores all edges incident on $v$. 
%%Throughout the algorithm, we never use these machines for any purpose other than storing the current graph on which $\palg$ is supposed to run. 
%%We say that a vertex (resp. an edge)  is \emph{active} if it belongs to the current graph. We maintain the invariant that for any $v$, machine $\MV{v}$ knows whether $v$ is active and if so, what are its active edges
%%(we show how to maintain this invariant in the following). With this in mind, we are now ready to present our main lemma of this section. 

\begin{lemma}\label{lem:mpc-parallel-O(1)}
	For a given graph $G(V,E)$, one can implement the following algorithms in the MPC model with at most $O(n)$ machines with memory $\Ot(n)$ with probability $1-1/n^{4}$: 
	\begin{enumerate}
		\item\label{item:mpc-1} Each call to $\pedcs$ in $\palg(G,n)$ in $O(1)$ MPC rounds. 
		\item\label{item:mpc-2} Each call to $\rmatch$ in $\palg(G,n)$ in $O(1)$ MPC rounds.
		\item\label{item:mpc-3} $\palg(G,n)$ in $O(\log\log{n})$ MPC rounds.
	\end{enumerate}
\end{lemma}

We prove each part of this lemma separately. 

\paragraph{Implementation of \pedcs.} To perform the vertex sampling approach in $\pedcs$, we need to annotate each edge with the subgraph(s) it is assigned to. To do this, each vertex $v$ in the current graph only needs to specify which subgraphs 
it resides on and broadcast this to its adjacent edges. Recall that unlike in~\cite{CzumajLMMOS17}, in our way of vertex sampling, each vertex can resides in multiple subgraphs (up to $\Ot(\sqrt{n})$ ones). Broadcasting this amount of information 
directly to adjacent edges of each vertex is not possible within the memory restrictions of the MPC model. However, recall that we are using an $O(\log{n})$-wise independent hash function for determining the subgraphs each vertex $v$ is going to reside on.
Hence, the vertex $v$ only needs to broadcast this hash function to its adjacent edges which requires $\polylog{(n)}$ bits for representation (see,~e.g.~\cite{RAbook}) and thus can be done in $O(1)$ MPC rounds on machines of memory $n^{\Omega(1)}$. 

We then send all edges assigned to one subgraph to a dedicated machine. By Claim~\ref{clm:mpc-pedcs-size}, the number of edges assigned to each machine is at most $\Ot(n)$ with high probability and hence it can fit the memory of the machine. 
We can then locally compute an EDCS of this subgraph and annotate the edges in this EDCS as the edges of the final multigraph $C$. All this can be easily done in $O(1)$ MPC rounds, hence finalizing this part of the proof. 

\paragraph{Implementation of \rmatch.} Each vertex in $S'$ simply needs to annotate one of its edges
uniformly at random, and each annotated edge only needs to ``mark'' its other endpoint in $G\setminus S'$. Any vertex in $G \setminus S'$ which is marked exactly once then inform the edge that marked it to join the matching $M$. 
This part can again be done in only $O(1)$ rounds on machines with memory $n^{\Omega(1)}$. 

\paragraph{Implementation of \palg.} We can now combine the results in the previous two parts to show how to implement $\palg$ in the MPC model. Consider a step of $\palg$. 
We saw that $\pedcs$ and $\rmatch$ can both be implemented in $O(1)$ MPC models. In particular, all edges in subgraph $C$ computed by $\pedcs$ are now annotated and hence we can ignore all remaining edges. We can also compute
the degree of each vertex in this subgraph in $O(1)$ rounds using a simple sort and search technique (see Lemma 6.1 in~\cite{CzumajLMMOS17}). We can hence compute the set of vertices $\vh$ and pass it to $\rmatch$ as the set $S$.
Finally, we can mark vertices in $\vh \cup V(\mh)$ and remove them from the graph (by broadcasting this information to all their neighbors).  After this, we know which vertices belong to $\cm$ for the next step
and which edges are still active. We can hence recursively solve the problem on the graph $\cm$ in the next steps. As each step requires $O(1)$ MPC rounds and there are $O(\log\log{n})$ steps in total by Claim~\ref{clm:palg-steps}, this results in an MPC algorithm with $O(\log\log{n})$ rounds.

\subsubsection{Optimizing the Number of Machines}\label{sec:mpc-palg-machines}

We conclude this section by making a remark about optimizing the number of machines (in addition to their memory) in our results as well. 

As it is, the total number of machines needed to implement $\palg$ in the MPC model is $\Ot(n)$. This means that the total memory across all machines is $\Ot(n^2)$, which is proportional to the input size (up to $\polylog{(n)}$ factors)
whenever the input graph is completely dense, i.e., has $O(n^2)$ edges. However for sparser graphs with $n^{1+\Omega(1)}$, this can be larger than the input size by a factor of $n^{1-\Omega(1)}$. 
This is consistent with some definitions of 	MapReduce-style computation such as~\cite{KarloffSV10,LattanziMSV11} but not with the strictest definitions in~\cite{AndoniNOY14,BeameKS13}, which require that the total memory of the system for a graph with $m$ edges to be only $\Ot(m)$, i.e., proportional to the input size. 

Nevertheless, a straightforward modification of our algorithm can reduce the number of machines down to $\Ot(m/n)$ which ensures that the total memory 
used by our algorithm is $\Ot(m)$ which adheres to the strictest restrictions of the MPC model. The only change we need to do is to work with the \emph{average degree} of the 
graph in $\palg$ as opposed to its maximum degree. Concretely, in each call to $\palg(G,\Delta)$, instead of computing $\pedcs(G,\Delta,n)$, we compute $\pedcs(G,\tDelta,n)$, where
$\tDelta:= m/n$ denotes the average degree of the graph $G$. It is easy to see that in this case, we still only need $\Ot(n)$ memory per machine (essentially the same argument in Claim~\ref{clm:mpc-pedcs-size} proves this), 
but now the total number of machines needed is only $\Ot(m/n)$ and hence we only need $\Ot(m)$ memory in total. It is easy to verify that the arguments in the proof of correctness of $\pedcs$ can be immediately 
applied to this version; we omit the details. 

\subsection{Extension to Smaller Memory Requirement and Proof of Theorem~\ref{thm:mpc-matching-vc}}\label{sec:mpc-smaller-memory}

As was shown by Lemma~\ref{lem:mpc-parallel-O(1)}, \palg can be implemented in the MPC model with machines of memory $\Ot(n)$ and $O(\log\log{n})$ MPC rounds. Combining this with Lemma~\ref{lem:palg-correct} on the correctness
of $\palg$, we immediately obtain Theorem~\ref{thm:mpc-matching-vc} for the case of $s = \Ot(n)$, i.e., an MPC algorithm with $O(1)$ approximation to both matching and vertex cover in $O(\log\log{n})$ rounds with $\Ot(n)$ memory per machine.

We now show how to extend our algorithm to the case when memory per machine is $s = n^{\Omega(1)}$. For simplicity, we assume the memory per each machine is $\Ot(s)$ as opposed to $O(s)$; rescaling the parameter 
$s$ with a $\polylog{(n)}$ factor implies the final result. Recall that there were only two places in $\palg$ that we needed $\Ot(n)$ memory per machine: in subroutine $\pedcs$ and in step $1$ of the algorithm, i.e., the base case of recursion. 
Consequently, we only need to make the following two changes to $\palg(G,\Delta)$: 

\begin{enumerate}
	\item Firstly, in $\palg(G,\Delta)$, we now run the subroutine $\pedcs$ with memory per machine equal to $O(s \cdot \polylog{(n)})$, i.e., we run $\pedcs(G,\Delta,s)$. 
	
	\item Second, we change the base case of the algorithm. Whenever $\Delta <  \paren{\frac{n}{s}} \cdot \paren{400 \cdot \log^{12}{(n)}}$, instead of sending all edges to a single machine and solve the problem locally, we 
	simply use any standard $O(\log{\Delta})$-round MPC algorithm for $O(1)$-approximation to matching and vertex cover that works on machines with memory $n^{\Omega(1)}$ (for example by directly simulating the
	distributed peeling algorithm of~\cite{OnakR10} (see also~\cite{ParnasR07}). For more details, see~\cite{CzumajLMMOS17} (Lemma~6.1).  
\end{enumerate}

Previously with machines of memory $\Ot(n)$, the maximum degree of underlying graph in each step of $\palg$ was (see Claim~\ref{clm:palg-pedcs}): 
\begin{align*}
	\Delta_1 : = n, ~~ \Delta_2 := \sqrt{n} \cdot \polylog{(n)}, ~~ \ldots ~~ \Delta_i := n^{1/2^{i}} \cdot \polylog{(n)}, ~~ \ldots ~~ \Delta_{T} := \polylog{(n)},  
\end{align*}
where $T = O(\log\log{n})$ denotes the number of steps in $\palg$. By switching to machines of memory $\Ot(s)$, we instead have, 
\begin{align*}
	\Delta_1 : = n, ~~ \Delta_2 := \frac{n}{s^{1/2}} \cdot \polylog{(n)}, ~ \ldots ~ \Delta_i := \frac{n}{s^{1-1/2^{i-1}}} \cdot \polylog{(n)}, ~ \ldots ~ \Delta_{T} := \paren{\frac{n}{s}} \cdot \polylog{(n)},
\end{align*}
before we reach the stopping condition of $\palg$ (this follows directly from Lemma~\ref{lem:pedcs-correctness} exactly as in Claim~\ref{clm:palg-pedcs}). After this step, we simply compute an 
$O(1)$-approximate matching and vertex cover directly in the remaining graph. 

The analysis of the correctness of this algorithm is exactly as before. It is also straightforward to verify that this algorithm now only needs
machines of memory $O(s \cdot \log^{2}(n))$ by choice of $\pedcs$. Finally, the number of rounds needed by this
algorithm is $O(\log\log{n})$ (for reducing the maximum degree in the graph to $\Delta_T$) plus $O(\log{\Delta_T}) = O(\log{\paren{\frac{n}{s}}} + \log\log{n})$ (for running the distributed matching and vertex cover algorithm directly when maximum degree is at most $\Delta_T$). This concludes the proof of Theorem~\ref{thm:mpc-matching-vc} for all range of parameter $s=n^{\Omega(1)}$.

\subsection{Further Improvements}\label{sec:mpc-wrap-up}

%%As was shown by Lemma~\ref{lem:mpc-parallel-O(1)}, \palg can be implemented in the MPC model with machines of memory $\Ot(n)$ and $O(\log\log{n})$ MPC rounds. Combining this with Lemma~\ref{lem:palg-correct} on the correctness
%%of $\palg$ algorithm, we immediately obtain Theorem~\ref{thm:mpc-matching-vc}, i.e., an MPC algorithm with $O(1)$ approximation to both matching and vertex cover. 

In the remainder of this section, we show that using standard techniques, one can improve the approximation ratio of our matching algorithm significantly. In particular, 

\begin{corollary}\label{cor:mpc-matching-2}
	There exists an MPC algorithm that given a graph $G$ and $\eps \in (0,1)$, with high probability computes a $(2+\eps)$-approximation to maximum matching of $G$ in $O(\log{(1/\eps)} \cdot {\log\log{n}})$ MPC rounds using only $O(n/\polylog{(n)})$ 
	memory per machine. 
\end{corollary}

\begin{corollary}\label{cor:mpc-matching-1}
	There exists an MPC algorithm that given a graph $G$ and $\eps \in (0,1)$, with high probability computes a $(1+\eps)$-approximation to the maximum matching of $G$ in $(1/\eps)^{O(1/\eps)} \cdot \paren{\log\log{n}}$
	 MPC rounds using only $O(n/\polylog{(n)})$ memory per machine. 
\end{corollary}

We note that above corollaries hold for all range of per machine memory $s= n^{\Omega(1)}$ similar to Theorem~\ref{thm:mpc-matching-vc}; however, for simplicity, we only consider the most interesting case of $s= O(n/\polylog{(n)})$.
We prove each of the above corollaries in the next two sections. 

\subsubsection{Proof of Corollary~\ref{cor:mpc-matching-2}}

The idea is to simply run our MPC algorithm in Theorem~\ref{thm:mpc-matching-vc}, to compute a matching $\ma$, remove all vertices matched by $\ma$ from the graph $G$, and repeat. Clearly, the set of all matchings 
computed like this is itself a matching of $G$. In the following, we show that only after $O(\log{1/\eps})$ repetition of this procedure, one obtains a $(2+\eps)$-approximation to the maximum matching of $G$.

Let $\alpha = O(1)$ be the approximation ratio of the algorithm in Theorem~\ref{thm:mpc-matching-vc}. 
Suppose we repeat the above process for $T := \paren{\alpha \cdot \log{(1/\eps)}}$ steps. For any $t \in [T]$, let $M_t$ be the matching computed so far, i.e., the union of the all the matchings in the first $t$ applications of our $\alpha$-approximation algorithm. 
Also let $G_{t+1} := G \setminus V(M_t)$, i.e., the graph remained after removing vertices matched by $M_t$. Note that $M_{t+1}$ is an $\alpha$-approximation to the maximum matching of $G_{t+1}$. 
Moreover, $\MM(G_{t+1}) \geq \MM(G) - 2\card{M_t}$ as each edge in $M_t$ can only match (and hence remove) two vertices of any maximum matching of $G$. This implies that 
$\card{M_{t+1}} \geq  \card{M_t} + \frac{1}{\alpha} \cdot \paren{\MM(G) - 2\card{M_t}}$ for all $t \in [T]$. We now have, 
\begin{align*}
	\MM(G) - 2\card{M_{T}} &\leq \paren{1-\frac{2}{\alpha}} \cdot \MM(G) - 2\cdot \paren{1-\frac{2}{\alpha}}\cdot \card{M_{T-1}} \\
	&= \paren{1-\frac{2}{\alpha}} \cdot \paren{\MM(G) - 2\cdot\card{M_{T-1}}} \\
	&\leq \paren{1-\frac{2}{\alpha}}^{2} \paren{\MM(G) - 2\cdot\card{M_{T-2}}} \tag{by applying the second equation to $M_{T-1}$} \\
	&\leq \paren{1-\frac{2}{\alpha}}^{T} \cdot {\MM(G)} \tag{by recursively applying the previous equation}  \\
	&\leq \exp\paren{-\frac{2}{\alpha} \cdot \alpha \cdot \log{(1/\eps)}} \cdot \MM(G)\leq \eps \cdot \MM(G).
\end{align*}

Hence, after $T = O(\log{1/\eps})$ steps, the matching computed by the above algorithm is of size $\paren{2+\eps} \cdot \MM(G)$. It is immediate to verify that the new algorithm can be implemented in the MPC model with 
machines of memory $O(n/\polylog{(n)})$ and $O(\log{(1/\eps)} \cdot \log\log{n})$ MPC rounds. Note that as the probability of error in the algorithm in Theorem~\ref{thm:mpc-matching-vc} is at most $1/n^{4}$, by a union bound, 
the new algorithm also outputs the correct answer with probability at least $1-1/n^{3}$. 

\subsubsection{Proof of Corollary~\ref{cor:mpc-matching-1}}\label{sec:mpc-matching-1}

Corollary~\ref{cor:mpc-matching-1} can be proven using Theorem~\ref{thm:mpc-matching-vc} plus a simple adaption of the multi-pass streaming algorithm of McGregor~\cite{McGregor05} for maximum matching to the MPC model. 
The high level approach in~\cite{McGregor05} is to reduce the problem of finding a $(1+\eps)$-approximate maximum matching
in $G$ to many instances of finding a maximal matching in multiple adaptively chosen subgraphs of $G$. It was then shown that there exists a single pass streaming algorithm that can both determine the appropriate subgraph of $G$ needed in each step
of this reduction and compute a maximal matching of this subgraph. Hence, after repeating this streaming algorithm in multiple passes over the stream, one can fully implement the reduction and obtain a $(1+\eps)$-approximation to the maximum matching. 

We show that essentially the same approach can also be used in the MPC model. 
The main difference is to switch from computing a maximal matching to finding an $O(1)$-approximate maximum matching using our Theorem~\ref{thm:mpc-matching-vc}. In the following, we briefly describe the approach
in~\cite{McGregor05} and point out the modifications needed to make it work in Corollary~\ref{cor:mpc-matching-1}. 
The purpose of this section is only to convince the reader that the reduction~\cite{McGregor05} can be seamlessly implemented in the MPC model and hence
we do not delve into the full details of the algorithm and analysis and instead refer the reader to~\cite{McGregor05} for more details and formal proofs. 

\paragraph{The Streaming Algorithm of~\cite{McGregor05}.}
The idea behind the algorithm is to start with some maximal matching $M$ (which is easy to compute in one pass over the stream) and then \emph{augment} this matching further over multiple \emph{phases} by finding 
a large set of vertex disjoint augmenting paths of length $O(1/\eps)$ in a \emph{layered} graph created from $G$ and the current matching $M$. We first introduce the concept of the layered graph that is used to reduce the task of finding 
augmenting paths to multiple instances of approximate matching. Given a graph $G$, a matching $M$, and an odd integer $k$ (which is the target length of the augmenting paths to be found), we create a graph $\LL(G,M,k)$ using
the following randomized procedure: 

\begin{tbox}
\begin{enumerate}
	\item The vertices in $\LL$ are the same as $G$ and are partitioned into $k+1$ subsets $L_1,\ldots,L_{k+1}$ called \emph{layers}. The layer of each vertex $v$ is determined as follows: 
	\begin{enumerate}
	\item For any vertex $u$ left unmatched by $M$, $u$ is assigned to either $L_1$ or $L_{k+1}$ chosen uniformly at random. 
	\item For any matched edge $(u,v) \in M$, the vertex $u$ is assigned to a uniformly at random chosen even layer $L_2,L_4,\ldots,L_{k-1}$ and $v$ is assigned to the subsequent layer. 
	\end{enumerate}
	\item For any edge $(u,v) \in {E \setminus M}$, if $u$ belongs to an odd layer and $v$ belongs to the next (even) layer then $(u,v)$ is added to $\LL$ as well. 
	\item For any edge $(u,v) \in M$, $(u,v)$ is added to $\LL$ as well. 
\end{enumerate}
\end{tbox}

One could easily verify that the edges in $\LL$ are only between two consecutive layers and for any edge in $\LL$ there is a unique edge in $G$. 
The main property of the above construction is that a collection of vertex disjoint paths between vertices in $L_1$ and $L_{k+1}$ corresponds to a set of vertex disjoint augmenting paths of length $k$ for $M$ in $G$. 
It is also relatively easy to prove that as long as $\card{M} < (1-\eps) \cdot \MM(G)$, then, there exists some $k = O(1/\eps)$, for which the corresponding graph $\LL(G,M,k)$ has at least $\MM(G) \cdot \eps^{O(1/\eps)}$ vertex disjoint paths 
between $L_1$ and $L_{k+1}$ with high probability (see Theorem~1 in~\cite{McGregor05}). 

We now describe how to find a large fraction of these augmenting paths in $\LL$ (in each phase) by finding different maximal matchings between consecutive layers of $\LL$. 
We first compute an approximate matching $M_{1,2}$ between $L_1$ and $L_2$. Let $L'_3 \subseteq L_3$ be the set of 
vertices that are neighbor to matched vertices of $M_{1,2}$ (in $L_2$). We then, compute a matching $M_{3,4}$ between $L'_3$ and $L_4$ and we continue like this. One can see this approach as growing vertex disjoint 
(augmenting) paths from layer $L_1$ to (eventually) layer $L_{k+1}$. If at some point, size of the matching between $L'_{i}$ and $L_{i+1}$ (for some odd $i$) goes below a certain threshold, we remove all vertices in $L'_{i}$ from the graph (as 
they are mostly ``dead ends'') and backtrack (similar to a DFS procedure). By choosing a relatively large threshold, one can ensure that the number of backtracks is bounded by some function of $(1/\eps)$ only and hence is not too large. At the same time, 
we like the threshold to be small enough also so that not many actual vertex disjoint paths are marked (incorrectly) as dead ends. Once we complete some paths from $L_1$ to $L_{k+1}$ we remove these paths (to be augmented later) and recurse on the 
remaining graph until no vertices are left. We again emphasize that at each step of this procedure, we simply find a maximal matching between some set of nodes in $\LL$ and the above algorithm determines which sets of vertices to choose for finding the next 
approximate matching. We refer the reader to Section 2.3 (in particular Fig 2) of~\cite{McGregor05} for a formal definition and a pseudo-code of this algorithm.

We now argue that essentially the same algorithm can also be implemented in the MPC model using our Theorem~\ref{thm:mpc-matching-vc} (instead of picking maximal matchings).

\paragraph{From Maximal to $O(1)$-Approximate Matching.} We point out that the algorithm of~\cite{McGregor05} uses a maximal matching in its construction as it is easy to compute in one pass over a stream and results in a $2$-approximate
matching. However, we emphasize that the analysis in~\cite{McGregor05} in no way uses the ``maximality'' property of this matching and only relies on its approximation ratio. 
We do not know how to compute a maximal matching in the MPC model efficiently, however,
we can use Theorem~\ref{thm:mpc-matching-vc} to compute an $O(1)$-approximate matching. By a simple adjustment of the thresholds used in the above algorithm, we can use any $O(1)$-approximation algorithm to maximum matching in place of a maximal
matching algorithm, while blowing up the number of calls to the matching subroutine between two layers by a constant factor in total. %Moreover, as each call to our algorithm requires $O(\log\log{n})$ M

\paragraph{From Streaming to MPC Model.} It is easy to see that the construction of the layered graph as well as the choices to which subgraph to compute the next matching on in the above algorithm can be easily performed in constant number
of rounds in the MPC model. As argued above, in place of the maximal matching algorithm in reduction of~\cite{McGregor05}, we use our algorithm in Theorem~\ref{thm:mpc-matching-vc} which requires $O(\log\log{n})$ MPC rounds (compared to one pass
in~\cite{McGregor05}). As a result, number of rounds in our algorithm is $O_{\eps}(\log\log{n})$ larger than the passes in the streaming algorithm of~\cite{McGregor05}. Overall, the new algorithm requires $(1/\eps)^{O(1/\eps)} \cdot O(\log\log{n})$ MPC 
rounds and $O(n/\polylog{(n)})$ memory per machine.

%% file: app-applications.tex
\section{Some Applications of Randomized Composable Coresets}\label{app:applications}
In the following, we provide more details on the applications of randomized coresets to different computational models and in particular prove Proposition~\ref{prop:coreset-application}.

%%Before that, we need 
%%to introduce a simple extension to the definition of randomized coresets. 
%%
%%Following~\cite{AssadiK17}, we augment the definition of randomized coresets by allowing the coresets to also contain a fixed solution (which is counted in the size of the coreset) to be directly added to the final solution of the composed coresets
%%(this is only needed for our vertex cover coreset; see~\cite{AssadiK17} on necessity of this definition for this problem). 
%%%In this case, size of the coreset is measured both in the number of edges in the output subgraph plus the number of vertices and edges picked by the fixed solution . 
%%%

\paragraph{MPC Algorithms.} The MPC model is defined in Section~\ref{sec:mpc}. We can use any $\alpha$-approximation randomized coreset \alg of size $s$ for a problem $P$ to obtain a parallel algorithm 
in only \emph{two} MPC rounds. Suppose $G(V,E)$ is the input graph and let $k:= \sqrt{m/s}$. The algorithm is as follows: 

\begin{tbox}\vspace{-0.25cm}
\begin{enumerate}
	\item In the first round, create a random $k$-partition $\Gi{1},\ldots,\Gi{k}$ and allocate each graph $\Gi{i}$ to the machine $i \in [k]$.  
	
	\item Each machine $i \in [k]$ creates a randomized composable coreset $C_i = \alg(\Gi{i})$. 
	
	\item In the second round, collect the union of all coresets to create $H:= H(V,C_1,\ldots,C_k)$ on one machine and solve $P$ in $H$ using any offline algorithm. 
\end{enumerate}
\end{tbox}

It is straightforward to verify that this MPC algorithm requires $O(k) = O(\sqrt{m/s})$ machines each with $O(\sqrt{ms} + n)$ memory. Moreover, by Definition~\ref{def:randomized-coreset}, the 
output of this algorithm is an $\alpha$-approximation to $P(G)$ with high probability. 

\paragraph{Streaming Algorithms.} In the streaming model, the edges of the input graph are presented to the algorithm one by one in a sequence and the algorithm is allowed to make a single pass (or a few passes) 
over this sequence. 

Similar to the case for MPC algorithms, any $\alpha$-approximation randomized coreset $\alg$ of size $s$ for a problem $P$ 
also imply a single-pass streaming algorithm for $P$ in \emph{random arrival} streams. Suppose $G(V,E)$ is the input graph which its edges are arriving in a random order in a stream of length $m$. 
The algorithms is as follows. 

\begin{tbox}\vspace{-0.25cm}
\begin{enumerate}
	\item Randomly partition the edges in the stream into $k := \sqrt{m/s}$ \emph{consecutive} pieces such that the distribution of the graphs $\Gi{1},\ldots,\Gi{k}$ created by edges in each piece
	is a random $k$-partition of $G$ (using the randomness in the arrival of the stream). 
	
	\item Let $H$ be the empty graph on vertices $V$. For $i = 1$ to $k$: 
	\begin{enumerate}
		\item Read the graph $\Gi{i}$ completely and store it in the memory temporarily. 
		\item Compute a randomized composable coreset $C_i = \alg(\Gi{i})$. 
		\item Update $H$ by adding all edges in $C_i$ to it and discard all edges in $\Gi{i}$ to reuse the memory in the next iteration.
	\end{enumerate}
	\item At the end of stream, solve $P$ in $H$ using any offline algorithm. 
\end{enumerate}
\end{tbox}

It is again easy to see that the total memory required by this algorithm is $O(m/k) = O(\sqrt{ms})$ (to store each graph $\Gi{i}$ in the memory temporarily) plus $O(k \cdot s) = O(\sqrt{ms} )$ (to store all coresets during the stream). 
The output of this algorithm is then an $\alpha$-approximation to $P(G)$ with high probability by Definition~\ref{def:randomized-coreset}. 

\paragraph{Simultaneous Communication Model.} In this model, the input graph is edge partitioned across $k$ machines/players and the goal is to solve the problem on the union of these graphs. 
In order to this, the players \emph{simultaneously} each send a single message to an additional party called the coordinator who then outputs the solution. 

Any $\alpha$-approximation randomized coreset $\alg$ of size $s$ for a problem $P$ 
immediately implies a simultaneous protocol for $P$ on \emph{randomly partitioned} inputs. Suppose $G(V,E)$ is the input graph which its edges are partitioned across $k$ parties \emph{randomly}. Each party only needs to compute
a coreset of its input graph and communicate it with the coordinator. The communication by each party then would be $O(s)$ and the coordinator can recover an $\alpha$-approximation to $P$ with high probability by Definition~\ref{def:randomized-coreset}.

\begin{remark}\label{rem:computational-time}
	In the above algorithms, we neglected the computation time needed for solving $P$ on the union of the coresets. This is consistent with the definitions of the models 
	considered here as they all allow unbounded computation time to the algorithm. However, if one insists on having efficient time algorithms (e.g., polynomial time) then
	the approximation ratio of the resulting algorithm would be $\rho \cdot \alpha$ where $\rho$ is the approximation guarantee of any offline algorithm we use for solving $P$ in the end (in many scenarios however this naive blow-up in the approximation
	ratio can be avoided by additional care, although not in a black-box way anymore). 
	\end{remark}

%% file: app-prelim.tex
\section{Missing Details from Section~\ref{sec:prelim}}\label{app:prelim}

\subsection{Useful Concentration of Measure Inequalities}\label{app:concentrations} 

 Azuma's inequality proves a concentration bound for martingales. 

\begin{proposition}[Azuma's inequality]\label{prop:azuma}
	Let $\set{X_i}_{i=0}^{n}$ be a martingale (with respect to some random variables $\set{Y_i}_{i=0}^{n}$). Suppose there exists a sequence of integers $\set{c_i}_{i=1}^{n}$ such that $\card{X_i - X_{i-1}} \leq c_i$; then, 
	\begin{align*}
		\Pr\Paren{\card{X_n - X_0} \geq \lambda} \leq 2\cdot \exp\paren{-\frac{\lambda^2}{\sum_{i=1}^{n} c_i^2}}.
	\end{align*} 
\end{proposition}

We also use the following standard variant of the Chernoff bound as well as its generalization for variables with bounded independence. 

\begin{proposition}[Chernoff bound]\label{prop:chernoff}
	Let $X_1,\ldots,X_n$ be independent random variables taking value in $[0,1]$ and $X := \sum_{i=1}^{n} X_i$. Then, for any $\delta \in (0,1)$
	\begin{align*}
		\Pr\Paren{\card{X - \Ex\bracket{X}} \geq \delta \cdot \Ex\bracket{X}} \leq 2\cdot \exp\paren{-\frac{\delta^2 \cdot \Ex\bracket{X}}{3}}. 
	\end{align*}
\end{proposition}

\begin{proposition}[Chernoff bound with bounded independence~\cite{SchmidtSS95}]\label{prop:bounded-chernoff}
	Let $X_1,\ldots,X_n$ be $\kappa$-wise independent variables taking value in $[0,1]$ and $X := \sum_{i=1}^{n} X_i$. For any $\delta \in (0,1)$, if $\kappa \leq \delta^2 \cdot \Ex\bracket{X}/2$, then, 
	\begin{align*}
		\Pr\Paren{\card{X - \Ex\bracket{X}} \geq \delta \cdot \Ex\bracket{X}} \leq 2\cdot \exp\paren{-\frac{\kappa}{2}}. 
	\end{align*}
\end{proposition}

%%\subsection{Missing Proofs of Section~\ref{sec:graph-facts}}\label{app:graph-facts}
%%
%%\begin{proposition*}
%%	Suppose $M$ and $V'$ are respectively, a matching and a vertex cover of a graph $G$ such that $\alpha \cdot \card{M} \geq \card{V'}$; then, both $M$ and $V'$ are $\alpha$-approximation 
%%	to their respective problems. 
%%\end{proposition*}
%%\begin{proof}
%%	$\VC(G) \Geq{Fact~\ref{fact:vc-matching}} \MM(G) \geq \card{M} \geq \frac{1}{\alpha} \cdot \card{V'} \geq \frac{1}{\alpha} \cdot \VC(G) \Geq{Fact~\ref{fact:vc-matching}} \frac{1}{\alpha} \cdot \MM(G)$. 
%%\end{proof}
%%
%%\begin{proposition*}
%%	Suppose $G(V,E)$ is a graph with maximum degree $\Delta$ and $\vh$ is the set of all vertices with degree at least $\gamma \cdot \Delta$ in $G$ for $\gamma \in (0,1)$. Then,  $\MM(G) \geq \frac{\gamma \cdot \Delta}{\Delta+1} \cdot \card{\vh}$. 
%%\end{proposition*}
%%\begin{proof}
%%	By Vizing's theorem~\cite{Vizing64}, $G$ can be edge colored by at most $\Delta + 1$ colors. As each color class forms a matching, this means that there exists a matching $M$ in $G$ with $\card{M} \geq \frac{\card{E}}{\Delta + 1}$. Moreover, 
%%	we have $\card{E} \geq \card{\vh} \cdot \gamma \cdot \Delta$, finalizing the proof. 
%%\end{proof}

\subsection{Proof of Lemma~\ref{lem:edcs-exists}}\label{app:edcs-exists}

\begin{lemma*}
	Any graph $G$ contains an $\EDCS{G,\beta,\beta^-}$ for any parameters $\beta > \beta^-$. 
\end{lemma*}
\begin{proof}
	Consider the following simple procedure for creating an EDCS $H$ of a given graph $G$: 
	\begin{tbox}
	\begin{enumerate}
		\item[] While $H$ is not an $\EDCS{G,\beta,\beta^-}$ of $G$: 
			\begin{enumerate}
				\item Find an edge $e$ which violates one of the properties of EDCS. 
				\item \emph{Fix} the edge $e$, i.e., remove it from $H$ if it was violating $\propone$ and add it to $H$ if it was violating $\proptwo$. 
			\end{enumerate}
	\end{enumerate}
	\end{tbox}
	The output of the above procedure is clearly an EDCS of graph $G$. However, a-priori it is not clear that this procedure ever terminates as fixing one edge $e$ can result in many edges violating the EDCS properties, potentially undoing
	the previous changes. In the following, we use a potential function argument to show that this procedure always terminates after a finite number of steps, hence implying that an $\EDCS{G,\beta,\beta-1}$ always exists. 
	
	We define the following potential function $\Phi$: 
	\[
		\Phi:= \paren{\beta - 1/2} \cdot \sum_{u \in V} \deg{H}{u} - \sum_{(u,v) \in H} \deg{H}{u} + \deg{H}{v}. 
	\]
	We argue that in any step of the procedure above, the value of $\Phi$ increases by at least $1$. Since the maximum value of $\Phi$ is at most $O(n \cdot \beta^2)$, this immediately implies that this procedure terminates in $O(n \cdot \beta^2)$ steps.
	
	Define $\Phi_1 := \paren{\beta - 1/2} \cdot \sum_{u \in V} \deg{H}{u}$ and $\Phi_2 := - \sum_{(u,v) \in H} \deg{H}{u} + \deg{H}{v}$ (note the minus sign) and hence $\Phi = \Phi_1 + \Phi_2$. 
	Let $(u,v)$ be the edge we choose to fix at this step, $H$ be the subgraph before fixing the edge $(u,v)$, and $H'$ be the resulting subgraph. 
	
	Suppose first that the edge $(u,v)$ was violating \propone of EDCS.  As the only change is in the degrees of vertices $u$ and $v$,
	$\Phi_1$ decreases by $(2\beta-1)$. On the other hand, $\deg{H}{u} + \deg{H}{v} \geq \beta+1$ originally (as $(u,v)$ was violating \propone of EDCS) and hence after removing $(u,v)$ $\Phi_2$
	increases by $\beta+1$. Additionally, for each neighbor $w$ of $u$ and $v$ in $H'$, after removing the edge $(u,v)$, $\deg{H'}{w}$ decreases by one. As there are at least 
	$\deg{H'}{u} + \deg{H'}{v} = \deg{H}{u} + \deg{H}{v} -2 \geq \beta - 1$ choices for $w$, 
	this means that in total, $\Phi_2$ increases by at least $(\beta + 1) + (\beta - 1) = 2\beta$. As a result, in this case $\Phi$ increases by at least $1$ after fixing the edge $(u,v)$. 
	
	Now suppose that the edge $(u,v)$ was violating \proptwo of EDCS instead. In this case, degree of vertices $u$ and $v$ both increase by one, hence $\Phi_1$ increases by $2\beta-1$. 
	Additionally, not that since edge $(u,v)$ was violating $\proptwo$ we have 
$\deg{H}{u} + \deg{H}{v} \leq \beta^- - 1$, so the addition of edge $(u,v)$ decreases $\Phi_2$ by 
at most $\deg{H'}{u} + \deg{H'}{v} = \deg{H}{u} + \deg{H}{v} + 2 \leq \beta^- + 1$.
	Moreover, for each neighbor $w$ of $u$ and $v$, after adding the edge $(u,v)$, $\deg{H'}{w}$ increases by one and since
	there are at most $\deg{H}{u} + \deg{H}{v} \leq \beta^- - 1$ choices for $w$, $\Phi_2$ decreases in total by at most $(\beta^- + 1) + (\beta^- - 1) = 2\beta^-$. 
	Since $\beta^-  \leq \beta-1$, we have that $\Phi$ increases by at least 
	$(2\beta-1) - (2\beta^-) \geq 1$ after fixing the edge $(u,v)$, finalizing the proof. 
\end{proof}

We remark that this proof also implies a natural polynomial time algorithm for computing any EDCS of a given graph $G$. 

%%\subsection{Proof of Lemma~\ref{lem:edcs-vc}}\label{app:edcs-vc}
%%\begin{lemma*}
%%	Let $G(V,E)$ be any graph, $\eps < 1/2$ be a parameter, and $H:= \EDCS{G,\beta,\beta^-}$ for parameters $\beta \geq \frac{4}{\eps}$ and $\beta^- = \beta \cdot \paren{1-\eps/4}$. 
%%	Suppose $\vh$ is the set of vertices $v \in V$ with $\deg{H}{v} \geq \beta^-/2$ and $\vvc$ is a minimum vertex cover of $H$; then $\vh \cup \vvc$ is a vertex cover of $G$ with size at most $(3+\eps) \cdot \VC(H)$ (note that $\VC(H) \leq \VC(G)$). 
%%\end{lemma*}
%%\begin{proof}
%%	We first argue that $\vh \cup \vvc$ is indeed a feasible vertex cover of $G$. To see this, notice that any edge $e \in H$ is covered by $\vvc$, and moreover by $\proptwo$ of EDCS, any edge $e \in E \setminus H$ has at least one endpoint with 
%%	degree at least $\beta^-/2$ in $H$ and hence is covered by $\vh$. In the following, we bound the size of $\vh$ by $(2+\eps)\cdot\VC(H)$, which finalizes the proof as clearly ${\card{\vvc}} = {\VC(H)}$. 
%%	
%%	By \propone of EDCS, the maximum degree of each vertex in $H$ is $\beta$. Moreover, for any vertex $v \in \vh$, we have $\deg{H}{v} \geq \beta^-/2 = (1-\eps/4) \cdot \beta/2$. 
%%	Hence, we can apply Proposition~\ref{prop:matching-high-degree} on the graph $H$ with parameters $D = \beta$ and $\gamma = \frac{1}{2} \cdot (1-\eps/4)$, and obtain, 
%%	\begin{align*}
%%		\VC(H) \Geq{Fact~\ref{fact:vc-matching}} \MM(H) \Geq{Proposition~\ref{prop:matching-high-degree}} \frac{1}{(2 + \eps)} \cdot \card{\vh}. 
%%	\end{align*}
%%	finalizing the proof. 
%%\end{proof}